\documentclass[letterpaper,twocolumn,10pt]{article}
\usepackage{arxiv}
\usepackage[utf8]{inputenc}
\usepackage[english]{babel}
\usepackage{subcaption}
\usepackage{float}
\usepackage{bbm}
\usepackage{enumitem}
\usepackage{url}
\usepackage{hyperref}
\usepackage{graphicx}
\usepackage{amsmath,amssymb,amsfonts}
\usepackage{amsthm}
\usepackage{algorithm}
\usepackage{algorithmicx}
 \usepackage{algpseudocode}
\usepackage{booktabs}
\usepackage{xspace}
\def\shownotes{1}       

\ifnum\shownotes=1
\newcommand{\authnote}[2]{{ $\ll$\textsf{\footnotesize #1 notes: #2}$\gg$}}
\else
\newcommand{\authnote}[2]{}
\fi

\providecommand{\vs}{vs. }
\providecommand{\ie}{\emph{i.e.,} }
\providecommand{\eg}{\emph{e.g.,} }

\providecommand{\myparab}[1]{\vspace{1pt}\noindent\textbf{#1} }

\expandafter\def\expandafter\normalsize\expandafter{%
    \normalsize%
    \setlength\abovedisplayskip{1pt}%
    \setlength\belowdisplayskip{2pt}%
    \setlength\abovedisplayshortskip{-2pt}%
    \setlength\belowdisplayshortskip{2pt}%
}
\usepackage{enumitem}

\providecommand{\sysname}{\textsc{Symphony}\xspace}
\providecommand{\sysnameslice}{\textsc{Symphony}\xspace}
\providecommand{\sysnamete}{\textsc{Symphony}\xspace}

\providecommand{\meta}{PWAN}

\providecommand{\blastshield}{\textsc{BlastShield}\xspace}
\providecommand{\simplex}{\textsc{LP-Simplex}\xspace}
\providecommand{\barrier}{\textsc{LP-Barrier}\xspace}
\providecommand{\scratch}{\textsc{LP-Reserved}\xspace}
\providecommand{\dote}{\textsc{DOTE}\xspace}
\providecommand{\oracle}{\textsc{Oracle}\xspace}

\newtheorem{proposition}{Proposition} 
\newtheorem{lemma}{Lemma} 
\newtheorem{theorem}{Theorem} 
\newtheorem{definition}{Definition}
\usepackage[subtle]{savetrees}
\usepackage[small,compact]{titlesec}
\usepackage{titling}
\newlist{compactitem}{itemize}{1}
\setlist[compactitem,1]{label=\textbullet, left=0pt, itemsep=1pt, topsep=1pt, parsep=0pt, partopsep=0pt}

\titlespacing\section{0pt}{12pt plus 2pt minus 2pt}{2pt plus 2pt minus 2pt}
\titlespacing\subsection{0pt}{12pt plus 2pt minus 2pt}{2pt plus 2pt minus 2pt}
\expandafter\def\expandafter\normalsize\expandafter{%
    \normalsize%
    \setlength\abovedisplayskip{-1pt}%
    \setlength\belowdisplayskip{-1pt}%
    \setlength\abovedisplayshortskip{-2pt}%
    \setlength\belowdisplayshortskip{2pt}%
}

\begin{document}
\date{}
\setlength{\droptitle}{-4em}   

\title{Stable and Fault-Tolerant Decentralized Traffic Engineering}
\author{
{\rm Arjun Devraj}\\
Cornell University
\and
{\rm Umesh Krishnaswamy}\\
Microsoft
\and
{\rm Ying Zhang}\\
Meta
\and 
{\rm Karuna Grewal}\\
Cornell University
\and
{\rm Justin Hsu}\\
Cornell University
\and
{\rm \'{E}va Tardos}\\
Cornell University
\and
{\rm Rachee Singh}\\
Cornell University
}
\maketitle

\begin{abstract}
Cloud providers have recently decentralized their wide-area network traffic engineering (TE) systems to contain the impact of TE controller failures. In the decentralized design, a controller fault only impacts its slice of the network, limiting the blast radius to a fraction of the network. However, we find that autonomous slice controllers can arrive at divergent traffic allocations that overload links by 30\% beyond their capacity. We present \sysname, a decentralized TE system that addresses the challenge of divergence-induced congestion while preserving the fault-isolation benefits of decentralization. By augmenting TE objectives with quadratic regularization, \sysname makes traffic allocations robust to demand perturbations, ensuring TE controllers naturally converge to compatible allocations without coordination. In parallel, \sysname's randomized slicing algorithm partitions the network to minimize blast radius by distributing critical traffic sources across slices, preventing any single failure from becoming catastrophic. These innovations work in tandem: regularization ensures algorithmic stability to traffic allocations while intelligent slicing provides architectural resilience in the network. Through extensive evaluation on cloud provider WANs, we show \sysname reduces divergence-induced congestion by $14\times$ and blast radius by 79\% compared to current practice.\footnote{Correspondence to: \href{mailto:adevraj@cs.cornell.edu}{adevraj@cs.cornell.edu}}
\end{abstract}

\section{Introduction}
The history of distributed systems teaches us a fundamental lesson: centralized control achieves optimal performance but creates unacceptable risk. Nowhere is this tension more acute than in planet-scale wide-area networks (WANs) that interconnect datacenters in the world's largest cloud deployments. For over a decade, cloud providers have relied on centralized software-defined traffic engineering (TE) to optimally route exabytes of traffic across continents~\cite{swan,b4,b4after}. However, this architectural choice harbors a massive risk: a single software controller fault can cascade into a planet-scale outage, taking down services for hundreds of millions of users~\cite{blastshield}.

\begin{figure}[t]
  \centering
    \begin{subfigure}[b]{0.24\textwidth}
    \includegraphics[width=\textwidth]{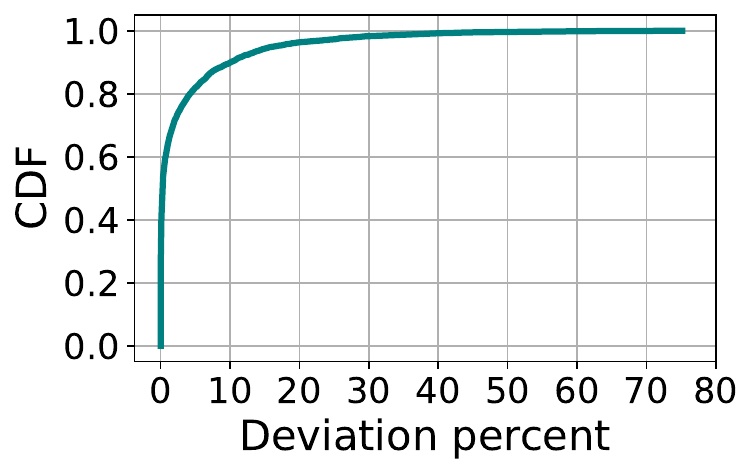}
    \caption{\small{Measured divergence.}}
    \label{fig:divergence}
  \end{subfigure}  
    \begin{subfigure}[b]{0.23\textwidth}
    \includegraphics[width=\textwidth]{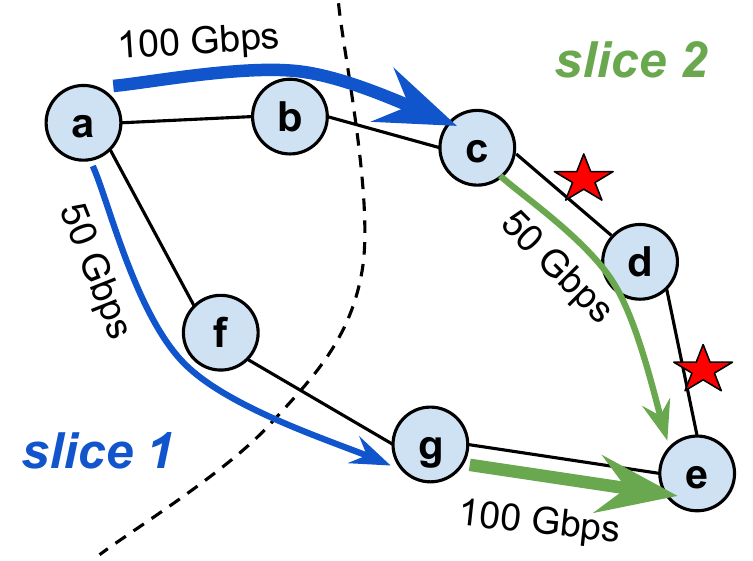}
    \caption{\small{Divergence example.}}
    \label{fig:slice-routing}
  \end{subfigure} 
  \vspace{1mm}
  \caption{\ref{fig:divergence} shows divergence in flow allocations between 2 pairs of slice controllers in a large commercial cloud WAN.~\ref{fig:slice-routing} shows an example of divergence in a network with 2 slices for demand $a\rightarrow e$. Due to differences in demands predicted by different controllers, slice 1's controller allocated 100 Gbps of traffic to the $abcde$ path and 50 Gbps to the $afge$ path. In contrast, slice 2's controller allocated 50 Gbps to $abcde$ and 100 Gbps to $afge$. Since each controller programs allocations in its slice, router $c$ receives 100 Gbps of traffic when it had only allocated 50 Gbps along link $cd$, causing congestion on starred links.}
\end{figure}

This realization has motivated an architectural change. Major cloud providers are increasingly \emph{decentralizing} their TE systems by breaking their WANs into autonomous slices, each governed by an independent controller~\cite{blastshield,onewan,meta-ebb}. In a decentralized WAN, a controller fault only impacts its slice, containing the blast radius to a fraction of the network. To achieve true fault isolation, each controller must operate independently, deriving its own view of network topology and traffic demands from independent measurements, without any coordination with other slices~\cite{onewan}.

However, decentralization introduces a new challenge. When slice controllers independently predict traffic demands from noisy switch counters, they inevitably arrive at slightly different demand matrices. These seemingly innocuous perturbations trigger a cascade of problems. Linear programs, the workhorses of modern TE, are notoriously sensitive to input variations. Small differences in predicted demands can cause controllers to compute vastly different routing decisions. Figure~\ref{fig:divergence} shows this challenge in a large production cloud WAN: in over 10\% of time windows, the bandwidth allocated by different slice controllers to the \emph{same path} diverges by more than 10\%. When upstream routers in one slice send traffic according to their controller's allocation, but downstream routers in another slice provision bandwidth according to a different allocation, inter-datacenter links become overloaded, resulting in congestion (Figure~\ref{fig:slice-routing}). Our findings show this \emph{divergence-induced congestion} can exceed link capacities by over 30\%, threatening the very reliability that decentralization was meant to achieve.

The obvious solutions fail to address this challenge. Eliminating noise in demand prediction is infeasible since it requires aggregating counters from thousands of distributed routers under tight time constraints. Further, forcing controllers to achieve consensus on demands undermines the autonomy that makes slices fault-tolerant by introducing the risk of globally propagating faults and errors in demand prediction. Finally, while operators could reserve 30\% spare capacity on every link to absorb divergence effects, this is both prohibitively expensive and scales poorly with network size. 
As a result, production TE systems face a stark choice: abandon decentralization or accept unpredictable congestion.

In this paper, we introduce \sysname, a novel approach to decentralized TE that reconciles the tradeoff between fault isolation and stable flow allocations. Our approach draws inspiration from an unexpected source: regularization in machine learning. Just as L2 regularization stabilizes neural networks from overfitting to noisy training data, we show that adding a quadratic regularizer to TE objectives makes routing decisions more robust to demand perturbations. However, regularization transforms TE into a quadratic program, threatening the scalability of TE formulations. We overcome this challenge through a key insight: not all network paths are equally susceptible to divergence. By selectively regularizing only the critical paths, \sysname maintains a similar computational complexity as traditional TE while scaling to the largest known WAN topologies. This paper makes three contributions that fundamentally reshape WAN traffic engineering:

\begin{compactitem}
    \item \textbf{Regularized TE:} We introduce a novel approach to TE that uses regularization to stabilize flow allocations in decentralized WANs. Through extensive evaluation on production cloud WANs, we demonstrate that \sysname reduces divergence-induced congestion by 14$\times$, while preserving the optimality of the original TE objective \emph{and} without requiring any coordination between slice controllers. 

    \item \textbf{Provable guarantees:} Our theoretical analysis demonstrates that \sysname provably reduces divergence in flow allocations regardless of the TE objective or slicing strategy.

    \item \textbf{Randomized slicing algorithm:} Finally, we show that naive WAN slicing strategies in decentralized TE can undermine fault isolation by concentrating high-volume traffic sources within the same slice. Thus, we develop a novel randomized algorithm that constructs network slices to minimize the blast radius by design. Our approach reduces the blast radius by 79\% compared to current practice while maintaining the stability benefits of regularization in TE.
    
\end{compactitem} 

Together, these contributions help mitigate the tension between divergence and fault isolation in planet-scale cloud traffic engineering systems. Moreover, \sysname{'s} implications extend beyond decentralized WANs, as even centralized TE systems struggle with demand prediction errors~\cite{dote}. While prior work tackles this challenge with deep learning~\cite{dote,valadarsky2017learning,teal}, these approaches remain under-deployed due to poor generalization and limited interpretability. In contrast, \sysname offers a theoretically grounded, immediately deployable solution to reduce sensitivity to demand prediction error, while also directly optimizing for operationally relevant congestion metrics.

\section{Divergence in Decentralized Network TE}
\label{sec:motivation}

Cloud WANs form the backbone of global digital infrastructure, carrying exabytes of traffic between datacenters distributed across the world. To manage this scale, cloud providers have relied on centralized software-defined traffic engineering (TE) systems to compute efficient routing decisions~\cite{swan,b4,te-survey}.

Centralized software-defined TE employs a single global controller that maintains complete knowledge of network state. In cloud WANs, the controller periodically collects two inputs: the network topology (\ie physical connectivity and link capacities) and the traffic demand matrix (\ie pairwise traffic volumes between datacenters). These inputs are then used in TE to compute optimal flow allocations along network paths, a task that requires solving a multi-commodity flow problem~\cite{te-survey}. In practice, operators formulate TE as a linear program (LP)~\cite{b4,swan} and solve it using commercial solvers such as Gurobi or CPLEX~\cite{gurobi,glop,cplex,glpk}.

Centralized TE enables operators to compute globally optimal flow allocations~\cite{swan,b4}, and solver-guided implementations of TE allow operators to encode diverse objectives tailored to their needs, such as maximizing total network throughput (MT), minimizing the maximum link utilization (MLU) to prevent congestion, or maximizing concurrent flow (MCF) to ensure fairness among competing demands~\cite{te-survey}. Every 5--10 minutes, the controller recomputes traffic allocations by solving the linear program with the updated inputs and reprograms forwarding tables across hundreds of routers, maintaining optimal network utilization even as traffic demands change~\cite{b4after,swan}.

However, centralized TE has a global blast radius: a single controller fault can disrupt the entire global network. Such failures have occurred in production cloud WANs, taking down services for hundreds of millions of users and prompting cloud providers to reconsider this architecture~\cite{blastshield}.

\begin{figure}
    \centering
   \includegraphics[width=0.99\linewidth]{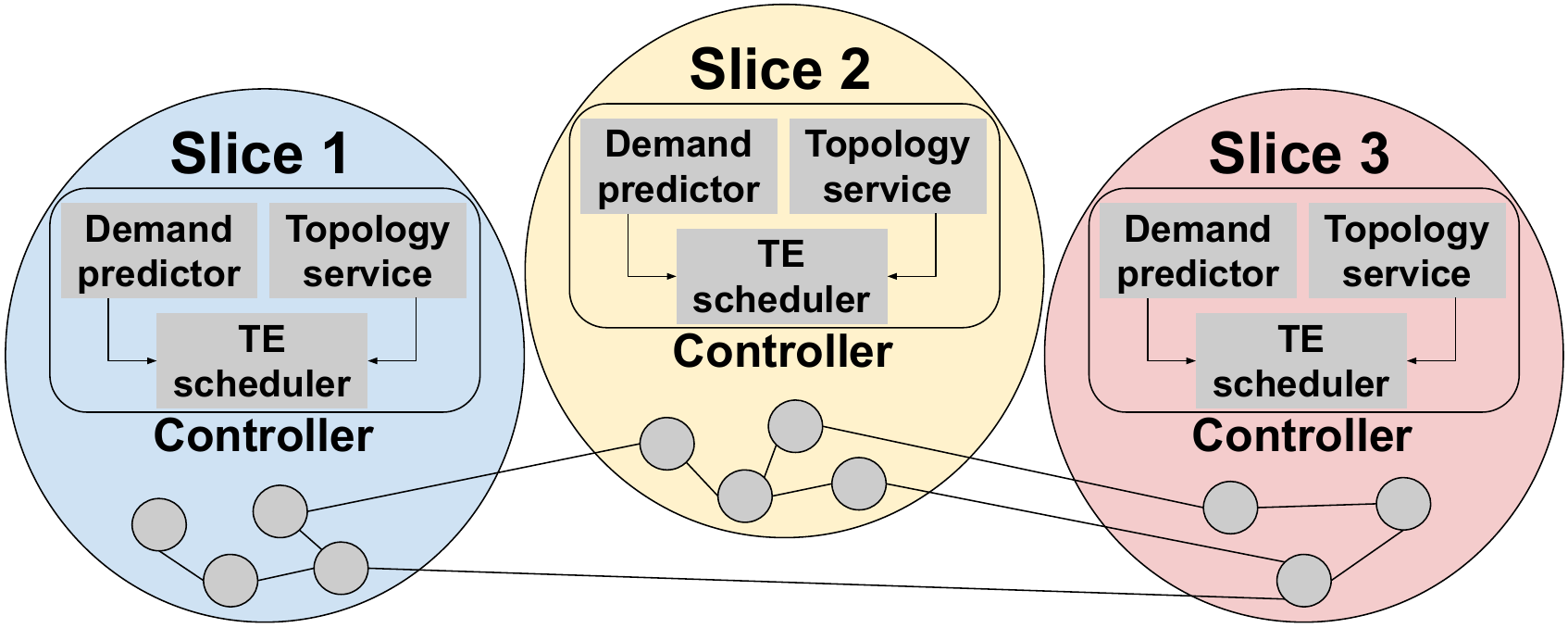}
   \vspace{0.5em}
    \caption{Architecture of decentralized TE systems.}     
    \label{fig:decentralized-te}
\end{figure}

\subsection{Architecture of Decentralized TE}
In response, cloud providers have implemented decentralized TE systems to limit fault propagation~\cite{blastshield,onewan,meta-ebb,dsdn}. Decentralized TE architectures partition the WAN into autonomous \emph{slices}---subgraphs of the network topology, each managed by an independent controller (Figure~\ref{fig:decentralized-te})~\cite{blastshield,onewan,meta-ebb}. When a controller experiences a fault, only the routers within its slice are affected, limiting the blast radius to slice boundaries. Complete autonomy between slices is critical for this guarantee, so controllers cannot share state or coordinate decisions. Each slice controller must independently:

\begin{compactitem}
\item \textbf{Build network topology} through its own \emph{topology service}, which polls routers to discover link states and then constructs the global network graph.
\item \textbf{Predict traffic demands} using its own \emph{demand predictor} service, which collects packet counters from switches and uses modeling techniques to estimate future traffic matrices.
\item \textbf{Solve the TE optimization} using its locally derived inputs to compute flow allocations on the global network.
\item \textbf{Program the forwarding state} in routers, either within its own slice (slice routing)~\cite{blastshield,onewan} or at traffic sources using segment (\ie source) routing protocols~\cite{dsdn}.
\end{compactitem}

In theory, this design is compelling: if all controllers observe identical network conditions (\ie topology, demands) and solve identical traffic allocation problems, they should compute identical routing decisions. Subsequently, each slice controller will program the forwarding state of routers (or traffic sources, with segment routing) in its own slice, but the network would behave identically whether controlled by a single global controller or multiple autonomous slice controllers. In this idealistic scenario, decentralization would provide fault isolation without sacrificing optimal routing~\cite{blastshield}.

\subsection{The Reality of Distributed Demand Prediction}
However, the distributed measurement of TE inputs in planet-scale networks is inherently noisy. Each slice controller must independently predict traffic demands for the next time window. A controller, using its demand predictor service (Figure~\ref{fig:decentralized-te}), must poll hundreds of routers, each maintaining millions of per-flow counters. These counters track packets traversing the network at line rates approaching terabits per second. The collection process faces numerous sources of error:

\begin{compactitem}
\item \myparab{Temporal skew.} Flow counter snapshots arrive at slice controllers with different delays due to variable network latency and processing time. Even with protocols like NTP, clock differences across routers contribute to the skew.
\item \myparab{Sampling artifacts.} Under load, routers may drop counter updates or employ sampling that introduces statistical noise.
\item \myparab{Prediction noise.} Forecasting models that extrapolate future demands from historical data can amplify input errors.
\end{compactitem}

\begin{figure}[t]
    \centering
    \includegraphics[width=0.6\linewidth]{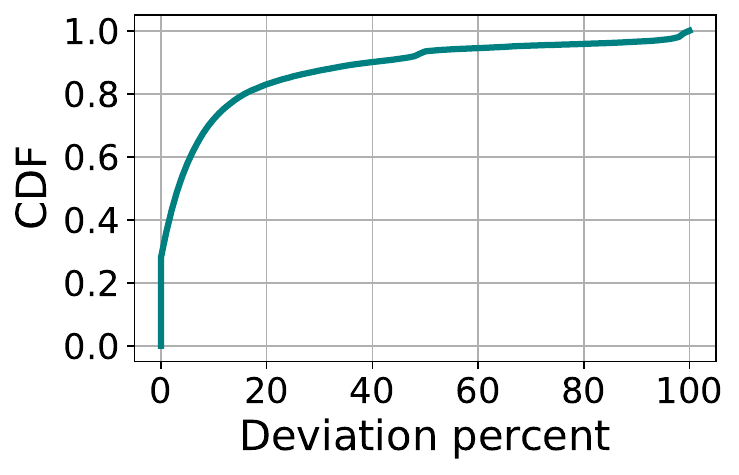}
    \caption{\small{Differences in demand inputs across 6 slice-controller pairs.}}
    \label{fig:perturbation}
\end{figure}

As a result, even when slice controllers measure the same network during the same time window, they inevitably derive slightly different demand matrices. Figure~\ref{fig:perturbation} quantifies these effects using measurements from a large commercial cloud WAN. Our data reveals that demand matrices computed by different slice controllers for the \emph{same time window} differ by over 10\% for 25\% of all source-destination pairs. These are not software bugs or configuration errors, but rather, represent the inherent limits of measurement accuracy in networks operating at planetary scale and petabit speeds.

\subsection{The Consequences: Divergence-Induced Congestion}
The impact of measurement noise is amplified by TE's formulation as a linear program~\cite{b4,b4after,swan,blastshield,onewan,dsdn}. While LPs are computationally efficient and guarantee optimal solutions, they are sensitive to input perturbations. This sensitivity is not an artifact of a specific implementation, but a fundamental property of linear optimization. When the optimal solution lies at the intersection of multiple constraints---a common occurrence in highly utilized networks---small changes in the objective coefficients or constraint bounds can cause the optimal solution vector to jump to an entirely different vertex of the feasible region. Thus, even minute differences in demand matrices used by different slice controllers can lead to drastically different flow allocations.

Figure~\ref{fig:divergence} demonstrates this phenomenon empirically. Slice controllers, starting with slightly different demand matrices, compute routing allocations that diverge by almost 80\%. As a result, Controller 1 might route 100 Gbps through the northern transcontinental path while Controller 2 routes the same traffic through the southern path. Both solutions are mathematically optimal for their respective inputs, yet they are fundamentally incompatible when implemented simultaneously, as is done in state-of-the-art decentralized TE deployments~\cite{blastshield,onewan}. Further, the unpredictability of divergence means that whether a path experiences divergence depends on the complex interaction between demand perturbations, network topology, and the current operating point of the optimization.

When divergent traffic allocations computed by different slice controllers manifest in the data plane, they create network congestion due to incompatible routing decisions. We call this \emph{divergence-induced congestion}, a problem that occurs regardless of the forwarding mechanism employed by the network (\eg segment routing~\cite{dsdn} or slice routing~\cite{blastshield,onewan}). Our findings show that divergence-induced congestion can \emph{unpredictably} overload links by as much as 30\% beyond their physical capacity, posing a significant challenge for WAN stability.

\subsection{The Failure of Conventional Solutions}
Traditional approaches to handling network uncertainty and preventing congestion prove inadequate against divergence-induced congestion.

\myparab{Coordinating input demands across slices} seems like an attractive solution: if all controllers agree on inputs, they should compute identical traffic allocations. However, this requires coordination protocols that compromise fault isolation. Consensus protocols like Paxos or Raft create dependencies between slices, as a Byzantine failure in one demand predictor could poison the agreed-upon demands, affecting all slices. Moreover, reaching consensus on millions of flow counters every few minutes across a planet-scale network introduces latencies that exceed the time budget for TE computation.

\myparab{Naively adjusting flow allocations at slice boundaries} will inherently drop traffic and lead to suboptimal network utilization, as downstream routers in a different slice may have either over- or under-allocated bandwidth for a given network path.

\myparab{Demand-oblivious routing} sidesteps prediction entirely by computing routes that work well for all possible demands~\cite{cope,texcp,oblivious-routing,oblivious-routing2,smore-nsdi}. These algorithms provide theoretical guarantees for worst-case performance regardless of actual demands. However, this robustness comes at an unacceptable cost: demand-oblivious schemes typically achieve orders of magnitude lower throughput than demand-aware routing, an unacceptable cost for many cloud providers~\cite{dote}.

\myparab{Reserving scratch capacity} could absorb the effects of divergence by ensuring links have sufficient headroom to handle conflicting routing decisions~\cite{swan}. Our measurements indicate this would require 30\% spare capacity---not just on specific links, but on \emph{all} links since divergence can occur unpredictably based on the complex interaction of noisy demands and optimization.

\myparab{Improving demand prediction} is an active area of research, with machine learning models that learn traffic patterns, smoothing noisy measurements, and ensemble methods that combine multiple predictors~\cite{dote,valadarsky2017learning}.  However, noisy demands are not just an artifact of modeling, but a fundamental property of the distributed measurement process itself. When routers drop counter updates under load and switch measurements inevitably vary even over small timescales, algorithmic sophistication cannot fully recover the missing information.

\myparab{Our core insight} is that noisy demand inputs in decentralized TE are not a byproduct of implementation choices, but an inevitable consequence of distributed measurement. Traditional networking approaches, which rely on linear programming formulations of TE to decide flow allocations, provide no inherent stability against such noise. As a result, small perturbations in demand estimates can lead to largely divergent flow allocations. Rather than attempting to eliminate noise—--which is unavoidable at WAN scale—--we instead reformulate the TE problem itself to be robust to perturbations, drawing on ideas from fields that have long confronted optimization under uncertainty.
\section{\sysname: Stable Decentralized TE}
\label{sec:design}

Divergence in decentralized TE occurs when flow allocations computed by different slice controllers—--even using
the same optimization formulation and software—--differ significantly due to perturbations in their predicted demands. 
TE is typically modeled as a linear program (LP) that optimally allocates network flows subject to capacity constraints.
Although LPs are computationally efficient, their solutions are highly sensitive to input perturbations.
The feasible region of an LP is a polytope, a convex geometric shape with flat faces and sharp vertices, whose optima lie at vertices. 
Standard optimization algorithms employed by LP solvers (\eg CPLEX, Gurobi), such as Simplex, may traverse completely 
different pivoting paths under slight demand perturbations, converging to entirely distinct solutions~\cite{networkflow}.

This challenge is not specific to the optimization algorithm (\ie Simplex \vs interior-point methods), but a fundamental property of linear programs, which can 
have multiple optimal solutions~\cite{boyd2004convex}. With no stable convergence point, even small demand perturbations can lead to 
dramatically different allocations, as we discuss in \S\ref{sec:motivation}. Because measurement noise is unavoidable, 
stability must come from the formulation itself: the TE objective should ensure that small input changes yield proportionally small changes in the solution vector.

\myparab{Our vision.} Rather than viewing divergence as an inevitable consequence of decentralized TE, \sysname elevates stability to a first-class design objective, on par with traditional network optimization goals such as throughput and fairness. Our approach combines two complementary innovations:

\begin{compactitem}
\item \myparab{Algorithmic stability through regularization (\S\ref{sec:design}):} In statistical learning theory, regularization prevents overfitting by penalizing models that are overly sensitive to training data~\cite{tikhonov1963solution, hinton1986learning, vapnik1998statistical}. We apply the same principle to TE, by using regularization to penalize routing solutions that are overly sensitive to demand perturbations. \sysname augments any existing convex TE objective with an L2 (quadratic) regularization term. This quadratic term fundamentally changes the optimization geometry: instead of a piecewise-linear objective that attains optima at sharp vertices of the feasible region, the problem becomes strictly convex with smooth, curved level sets. The result is a unique global optimum---a single stable point toward which all controllers converge, even under slightly different inputs. Small demand perturbations shift the optimum continuously rather than causing abrupt jumps. In \S\ref{sec:proof}, we prove that L2 regularization reduces divergence, regardless of the TE objective or traffic demands.
\item \myparab{Architectural resilience through intelligent slicing (\S\ref{sec:slicing}):} While regularization provides algorithmic stability, the physical partitioning of the network into slices determines fault resilience. Not all slicing strategies are equal: poor choices can concentrate traffic sources in a single slice, amplifying rather than containing failures. To address this, we develop a novel randomized algorithm to partition the WAN into slices that minimize the blast radius of controller faults. Our algorithm strategically distributes high-traffic nodes across slices to ensure balanced load and connectivity. By generating multiple candidate partitions that meet fault-tolerance goals, \sysname provides operators with flexibility to choose based on additional operational constraints like geography. Using data from a production cloud WAN, we show that our algorithm produces slices that reduce the blast radius by 79\% compared to the state-of-the-art, all while mitigating divergence.
\end{compactitem}

\sysname's regularizer preserves all existing TE objectives and constraints. Operators can continue to optimize for maximum throughput, minimum congestion, or fair allocation; \sysname simply makes these optimizations stable. Further, \sysname{} is immediately deployable, as operators only need to modify their optimization formulation, not their entire TE infrastructure (\eg measurement systems, routing protocols, switch hardware). As a result, slice controllers naturally converge toward similar routing decisions without global coordination, achieving stable, truly decentralized TE.

\subsection{Technical Challenges}
\sysname{'s} vision requires overcoming key technical challenges, which we solve with insights from studying WANs.

\begin{compactitem}

\item \myparab{Challenge 1: Scalability of quadratic TE.}
L2 regularization transforms the TE problem from a linear program (LP) to a quadratic program (QP). While regularized QPs are still convex and solvable in polynomial time, they typically take longer to solve than LPs. For planet-scale networks, this overhead could exceed the tight time budgets for TE computation. 

\myparab{Solution: Selective regularization.}
We observe that not all links are equally likely to experience divergence-induced congestion. Links that only carry traffic engineered by a single slice controller (as determined by the network topology) \emph{cannot} experience divergence. Further, we can use historical data to estimate the \emph{worst-case} load on each link; if this load is sufficiently below the link's capacity, divergence-induced congestion is highly unlikely. Using these insights, we selectively apply regularization only where needed. This keeps the problem size manageable while maintaining stability. Our techniques maintain equivalent runtime to baseline TE while scaling to the largest known WAN topologies (\S\ref{sec:scaling}).

\item \myparab{Challenge 2: Complexity of graph partitioning.} Graph partitioning is notoriously NP-hard, and adding fault-tolerance goals only further increases complexity. Naive approaches perform poorly. For instance, minimizing edge cuts between slices, a common partitioning objective, can group high-traffic nodes together and thus create slices with disproportionate load---exactly the opposite of fault tolerance. 

\myparab{Solution: Warm-start randomization.} Constructing slices that are locally connected and balanced in both resources (\ie nodes and routers) and traffic results in a large combinatorial search space. To efficiently solve this, we design a randomized algorithm 
that exploits a key property of WANs: traffic follows a heavy-tailed distribution where a few ``elephant'' sources generate most of the flow~\cite{demand-pinning,onewan}. By seeding slices with elephant sources, 
our algorithm substantially prunes the search space of suboptimal traffic-imbalanced slicing candidates. Then, it iteratively balances traffic and node counts across slices, using randomization to generate multiple candidates that operators can choose from based on operational constraints.
This prevents any slice from becoming a bottleneck while maintaining the connectivity needed for local control. We show that our algorithm generates slices that reduce the worst-case blast radius while keeping divergence minimal (\S\ref{sec:slicing}).

\end{compactitem}
\section{Regularization in traffic engineering}
\label{sec:design}

Preventing an ML model from \emph{overfitting} to data is well studied in the optimization literature. Applying L2 (Tikhonov) regularization to the objective function has greatly improved models' ability to generalize beyond the training data~\cite{tikhonov1963solution, hinton1986learning, vapnik1998statistical}. L2 regularization has primarily been studied in prediction problems that solve $\min_{\mathbf{w}} f(\mathbf{X}, \mathbf{y}, \mathbf{w})$ given 
training data $\mathbf{X}$, labels $\mathbf{y}$, and parameters/decision variables $\mathbf{w}$. However, with L2 regularization, the problem becomes
\begin{equation}
\min_{\mathbf{w}} f(\mathbf{X}, \mathbf{y}, \mathbf{w}) + \lambda \|\mathbf{w}\|_2^2
\end{equation}
where $\lambda$ is the regularization parameter (typically a small positive value, like 1e-4) and $\|\mathbf{w}\|_2^2$ is the L2-norm of $\mathbf{w}$ (\ie sum of squared weights). L2 regularization reduces overfitting because it penalizes large values of $w_{i}$, thereby reducing the impact of any one $w_{i}$ on the objective. 

The standard TE problem is fundamentally a linear program. The controller must decide how much bandwidth (\ie flow) to allocate to each network path in order to optimize for the desired objective, given a predicted demand matrix and link capacity constraints. 
In TE, the decision variables are $w_{p}$, the proportion of demand for a given flow that should be routed along path $p$. Applying L2 regularization to TE thus has the effect of \emph{load-balancing} traffic across paths for each flow. Further, the now \emph{quadratic} program (QP) encourages the optimizer to find interior solutions that are less affected by adjustments to the constraints (\eg demand and capacity inputs), and the strict convexity of the objective ensures that the solution is \emph{unique}, with sensitivity guarantees~\cite{boyd2004convex}.

\myparab{Beyond naive L2 regularization.} While L2 regularization can improve 
resilience to demand perturbations, operators are primarily interested in the \emph{congestion} induced by divergence. 
Standard L2 regularization encourages splitting traffic naively across paths, but disregards key topology-specific considerations. To combat divergence, we take inspiration from L2 regularization to design \sysname. 
\sysname penalizes all TE objectives with the term $\lambda \sum_{e \in E} u_{e}^{2}$, for links $e \in E$ and their utilizations $u_{e}$ (\ie $u_{e}=1$ means the link is operating under full load). Unlike a naive regularizer, penalizing by $u_{e}^{2}$ still encourages load-balancing traffic, but informs splitting ratios by prioritizing links based on their vulnerability to divergence-induced congestion. Across all slice solutions, links with high $u_{e}$ are the most likely to experience congestion when flow allocations diverge, while links with lower $u_{e}$ have more ``slack'' to tolerate divergence.

\myparab{Improvements for key objectives.} The three primary objectives in WAN traffic engineering are (1) maximizing throughput (MT)/minimizing unsatisfied demand, (2) maximizing concurrent flow (MCF), and (3) minimizing the maximum link utilization (MMLU). 
\sysname augments each of these objectives with L2 regularization. \sysname{'s} formulation for maximizing throughput is shown in Algorithm~\ref{alg:te}, whereas
we only show modifications to the objective function for MCF and MMLU. As is common in WAN TE, we write all formulations with decision variables as weights/splitting ratios for paths, as opposed to raw bandwidth allocations~\cite{dote, figret}. These are equivalent, as the allocation on path $p$ for flow $i$ with demand $d_{i}$ is $x_{p} = w_{p}d_{i}$. For all objectives, \sysname{} only modifies the objective function; constraints remain unchanged from standard LPs.

\begin{algorithm}[t]
    \begin{flushleft}
        \textbf{Inputs:}
    \end{flushleft}
    \begin{tabular}{p{2cm}p{6cm}}
        $G\langle V,E\rangle$ & network $G$ with vertices $V$ and links $E$\\
        $d_{i} \in D$ & traffic demand of source-destination pair $i$\\
        $P_{i}$ & set of paths for flow $i$\\
        $c_{e}$ & capacity of link $e$\\
        $I(p,e)$ & indicator variable for path $p$ using link $e$\\
        $\lambda$ & regularization parameter\\
    \end{tabular}
    \begin{flushleft}
        \textbf{Auxiliary Variables:}
    \end{flushleft}
    \vspace{-0.5em}
    \begin{tabular}{p{2cm}p{6cm}}
        $u_{e}$ & utilization of link $e$ \\
    \end{tabular}
    \begin{flushleft}
	    \textbf{Output:}
    \end{flushleft}
    \vspace{-0.5em}
    \begin{tabular}{p{2cm}p{6cm}}
        $w_{p}$ & flow splitting ratio along path $p$\\
     \end{tabular}
    \begin{flushleft}
        \textbf{Minimize} $\sum_{i} d_{i}(1 - \sum_{p \in P_{i}} w_{p}) + \lambda \sum_{e \in E} u_{e}^{2}$ \\[0.1em]
        \emph{subject to:}
    \end{flushleft}
    \begin{tabular}{ll}
        $\sum_{p \in P_{i}} w_{p} \leq 1$ & $\forall d_{i} \in D$\\[0.5em]
        $\sum_{i} d_{i} \sum_{p \in P_{i}} w_{p} I(p,e) \leq c_{e}$ & $\forall e \in E$\\[0.5em]
        $u_{e} = \frac{\sum_{i} d_{i} \sum_{p \in P_{i}} w_{p} I(p,e)}{c_{e}} $ & $\forall e \in E$\\[0.5em]
    \end{tabular}
\caption{\sysname{} Throughput Maximization}
\vspace{-0.5em}
\label{alg:te}
\end{algorithm}

\myparab{Maximizing throughput.} To align with the literature on regularization in ML prediction problems, we rewrite maximizing total flow as the equivalent problem of minimizing unsatisfied demand, subject to the same constraints. \sysname{} introduces the regularization penalty, $\lambda \sum_{e \in E} u_{e}^{2}$ (see Algorithm~\ref{alg:te}).

\myparab{Maximizing concurrent flow.} The MCF objective ensures fairness between flows by using the parameter $\gamma \in [0,1]$ to signify the proportion of requested demand each flow can send. While naive MCF maximizes $\gamma$, the \sysname{} objective now also incorporates the regularization penalty:
\begin{equation}
    \mathbf{maximize} \quad \gamma - \lambda \sum_{e \in E} u_{e}^{2}
\end{equation}
For fairness, the MCF demand constraint is $\sum_{p \in P_{i}} w_{p} = \gamma, \forall i$.

\myparab{Minimizing MLU.} The \sysname objective now becomes 
\begin{equation}
    \mathbf{minimize} \quad Z + \lambda \sum_{e \in E} u_{e}^{2},
\end{equation}
where $Z \geq u_{e}, \forall e \in E$, since $Z$ is the auxiliary variable for the maximum link utilization in the network. Unlike other objectives, which try to increase flow, MMLU aims to minimize congestion and thus lacks a capacity constraint to avoid scenarios in which no feasible solution is found (see \S\ref{subsec:mlu-formulation}).

\myparab{Why a quadratic penalty?} Consider the L1 regularizer, $\lambda \sum_{e \in E} u_{e}$, and a network with two links, $e1$ and $e2$. While allocations that result in $(u_{e1}, u_{e2})$ values of $(1,0)$ vs. $(0.5, 0.5)$ are equally optimal for the L1 regularizer, only the latter is optimal for the L2 regularizer. Not only does \sysname{} load-balance traffic to reduce the risk of divergence-induced congestion, but it also finds optimal solutions that are less sensitive to demand perturbations (\S\ref{sec:proof}).

\myparab{Choosing the regularization parameter.} While there is extensive research in machine learning about choosing $\lambda$~\cite{bergstra2012random}, we find that the choice of $\lambda$ is insignificant by several orders of magnitude in practice (see~\S\ref{subsec:appendix-reg-param}). The value of $\lambda$ can be fixed for each objective, enabling easy deployment without refresh.

\myparab{Routing protocol.}While \sysname is compatible with both slice and segment routing, we implement it with segment routing, as it is widely supported by switch vendors, commonly deployed in decentralized WANs~\cite{dsdn}, and agnostic to slicing.
\section{Provably reduced sensitivity to perturbations}
\label{sec:proof}
In addition to its intuitive benefits, \sysname{} also \emph{provably} reduces TE's sensitivity to demand perturbations. The formulations for all three standard TE objectives are already convex. Consider adding the \sysname{} regularizer:
\begin{equation}
    g(\mathbf{x}) = \lambda \sum_{e \in E} u_{e}^{2} = \lambda \sum_{e \in E} \Bigg(\frac{\sum_{i} \sum_{p \in P_{i}} \mathbf{x}_{p} I(p,e)}{c_{e}}\Bigg)^{2},
\end{equation}
where $\mathbf{x}$ is now the vector of path flow allocations (an equivalent formulation, as $x_{p} = w_{p} d_{i}$ from Algorithm~\ref{alg:te}).

Standard L2 regularization on the path variables benefits from strict convexity~\cite{boyd2004convex}, but introduces $O(k|V|^{2})$ quadratic terms to the objective (assuming $k$ paths per demand and $|V|$ nodes). In contrast, the \sysname{} regularizer is either naturally strictly convex, or can be easily made strictly convex, and bounds the number of quadratic terms in the objective at $O(2|E|) << O(k|V|^{2})$ while also directly targeting divergence-induced congestion. Adding the strictly convex $g(\mathbf{x})$ regularizer makes the objective \emph{strictly} convex, ensuring a \emph{unique} solution and enabling the optimal flow allocation to be \emph{Lipschitz continuous} with respect to input values. Therefore, the optimal solution, \ie path weights $w_{p}$, changes in a \emph{bounded} manner given demand perturbations, unlike conventional LPs that lack any such guarantees and are highly sensitive to demand values.

\subsection{Guaranteeing a unique solution}
Let $\mathbf{A}$ be the normalized edge-path incidence matrix, where $\mathbf{A}_{e, p} = I(p, e)/c_{e}$. Then, we can rewrite
\begin{equation}
  g(\mathbf{x}) 
  \;=\; 
  \lambda \sum_{e \in E}
  \Bigl(
    \sum_{p} x_p \,\frac{I(p,e)}{\,c_e\,}
  \Bigr)^2
  \;=\;
  \lambda \,\bigl\| \mathbf{A}\,\mathbf{x} \bigr\|_2^2.
\end{equation}

\begin{proposition}
The Hessian of $g(\mathbf{x})$ is
\begin{equation}
\label{eq:Hg}
\nabla^2 g(\mathbf{x})
\;=\;
2\,\lambda\,\mathbf{A}^\top \mathbf{A}.
\end{equation}
\end{proposition}

\begin{lemma}
\( g(\mathbf{x}) \) is strictly convex if $\mathbf{A}$ has full column rank.
\end{lemma}
\begin{proof}
For $g(\mathbf{x})$ to be strictly convex, $\nabla^2 g(\mathbf{x})$ must be positive definite. Thus, for any $\mathbf{z} \neq 0$ in $\mathbb{R}^{n}$, it must hold that \[
  \mathbf{z}^\top
  \bigl( \nabla^2 g(\mathbf{x}) \bigr) 
  \, \mathbf{z} = 2\lambda \, \mathbf{z}^\top 
  \bigl(\mathbf{A}^\top \mathbf{A}\bigr)
  \mathbf{z} > 0
\]
Since $\lambda > 0$, it follows that this expression is strictly positive if $\mathbf{A}$ has full column rank.
\end{proof}

If $\mathbf{A}$ is column-rank deficient, then it is sufficient to add $|E| - rank(\mathbf{A})$ phantom edge utilizations ($u_{e}'$) to the objective to ensure $\mathbf{A}$ has full column rank. These $u_{e}'$ simply bundle path weights in order to create $\mathbf{A}$ with full column rank, and can have arbitrarily large capacity and a much smaller value of $\lambda'$ to discount their impact. In practice, $u_{e}'$ are completely unnecessary, and our evaluation excludes them; further, path selection strategies can aim to ensure $\mathbf{A}$ is full rank without introducing any $u_{e}'$. Even with $u_{e}'$, the number of quadratic terms is still bounded in the number of links while ensuring sensitivity guarantees.

\subsection{Bounding the sensitivity of flow allocations}
Adding the strictly convex regularizer $g(\mathbf{x})$ to the convex objective makes the new objective  $f(\mathbf{x})$  strictly convex.
\begin{theorem}
If $f(\mathbf{x})$ is strictly convex, the optimal solution $\mathbf{x^{*}}$ is Lipschitz continuous with respect to the problem inputs $d_{i}$ and $c_{e}$:
$\|\mathbf{x^*(\mathbf{d_1})} - \mathbf{x^*}(\mathbf{d_2})\| \;\le\; L\,\|\mathbf{d}_1 - \mathbf{d}_2\|$ for all demand inputs $\mathbf{d_{1}}, \mathbf{d_{2}}$ and positive real constant $L$\text{~\cite{bonnans2013perturbation}}.
\end{theorem}
In contrast, LPs are not strictly convex, so they lack a unique solution and cannot bound changes in flow allocations under demand perturbations; even slight input changes can cause uncontrolled divergence. Unlike LPs, \sysname{} is strictly convex, ensuring a unique solution and bounded changes in flow allocations. This results in provably lower divergence while introducing fewer quadratic terms than naive L2 regularization.

\section{Scaling \sysnamete}
\label{sec:scaling}

Quadratic programs typically scale more poorly than linear programs in the number of variables and constraints. 
However, we design strategies that enable \sysnamete to solve as quickly as LPs on large networks. For the MT and MCF objectives, we improve the runtime by reformulating constraints to better leverage the newfound curvature of the objective function~\cite{boyd2004convex}. For MMLU, we selectively prune the number of quadratic terms from the objective without sacrificing performance by leveraging insights from source routing.

For objectives with a capacity constraint (MT, MCF), simply reformulating each capacity constraint as $u_{e} \leq 1$ improves runtime sufficiently. This increases sparsity, halving the number of non-zero terms in the KKT matrix, and halves presolve time.

\subsection{Pruning risk-free links}
To reduce \sysnamete{'s} runtime for the MMLU objective, we remove the quadratic penalty from the objective for links that are determined to be \emph{divergence-free}.
\begin{definition}
    A link \( e \) is divergence-free if the set of sources for all paths that use \( e \) 
    belong to the same slice. Let \(\mathcal{S}(p)\) be the slice that the source of path \( p \)
    belongs to. Then, \( e \) is divergence-free if: \(\forall p_1, p_2 \ni e,\) \(S(p_1) = S(p_2)\).
\end{definition}
With source routing, the source of truth for computing flow allocations 
only differs when the source nodes for each path belong 
to different slices. If all paths using a link $e$ are sending flow at rates that 
have been determined by the same slice controller, then $e$ cannot experience divergence.

\subsection{Pruning low-risk links}
We can also remove quadratic terms from \sysnamete{'s} MMLU objective 
by identifying links that are highly unlikely to suffer from divergence. Specifically, we can 
prune regularizers for links that meet the criterion of being $\beta$\emph{-constrained}.
\begin{definition}
    A link \( e \) is \( \beta \)-constrained if the sum of the 
    maximum past observed demand values for paths that use \( e \), divided by its capacity, 
    is at most \( \beta \). Formally, let \(\mathbbm{1}_{e}(i)\) be an indicator that is 1 if any path 
    for flow \( i \) uses edge \( e \). Let \( d_{i,t} \) denote the past demand 
    for flow \( i \) at time \( t \). Then, \( e \) is \( \beta \)-constrained if:
    \[
    \frac{1}{c_{e}} \sum_{i} \mathbbm{1}_{e}(i) \cdot \max_t d_{i,t} \leq \beta
    \]
\end{definition}

$\beta$-constrained is a highly conservative criterion for determining whether a link $e$
will ever suffer from divergence. It assumes the worst-case configuration 
given past demands and asks: if \emph{all} demands that have 
some path using link $e$ were to send their \emph{maximum} observed demand \emph{fully}
along any path that uses $e$, would the total flow on $e$ exceed $\beta \cdot c_{e}$?
Typically, values of $\beta \geq 1$ are sufficient given the risk-averse definition, 
while lower $\beta$ values are more conservative.
\section{Evaluating \sysnamete{}}
\label{sec:eval-divergence}

We evaluate \sysnamete{} on four WAN topologies (Table~\ref{tab:topologies}). We collect traffic demands for \meta{}, a large content provider WAN, over a 6-month period and evaluate \sysnamete{} on high-priority 
flows. Demand matrices for GEANT~\cite{uhlig2006providing} and ATT~\cite{teavar} are publicly available, while we simulate demands for KDL using the popular gravity model~\cite{zhang2003fast}. For ATT, GEANT, and \meta{}, 
we generate candidate partitionings using our algorithm in \S\ref{sec:slicing} and pick the slicing configuration with the lowest blast radius. 
For KDL, we generate a random slicing configuration. For all topologies, the slicing configuration aims for decentralization, with 
$k\approx\sqrt{n}$ slices, each with roughly $\sqrt{n}$ nodes (given $n$ nodes total). We also evaluate \sysnamete{} under different slicing configurations in \S\ref{sec:slicing}.

\begin{figure*}[t]
    \centering   
    \includegraphics[width=0.75\textwidth]{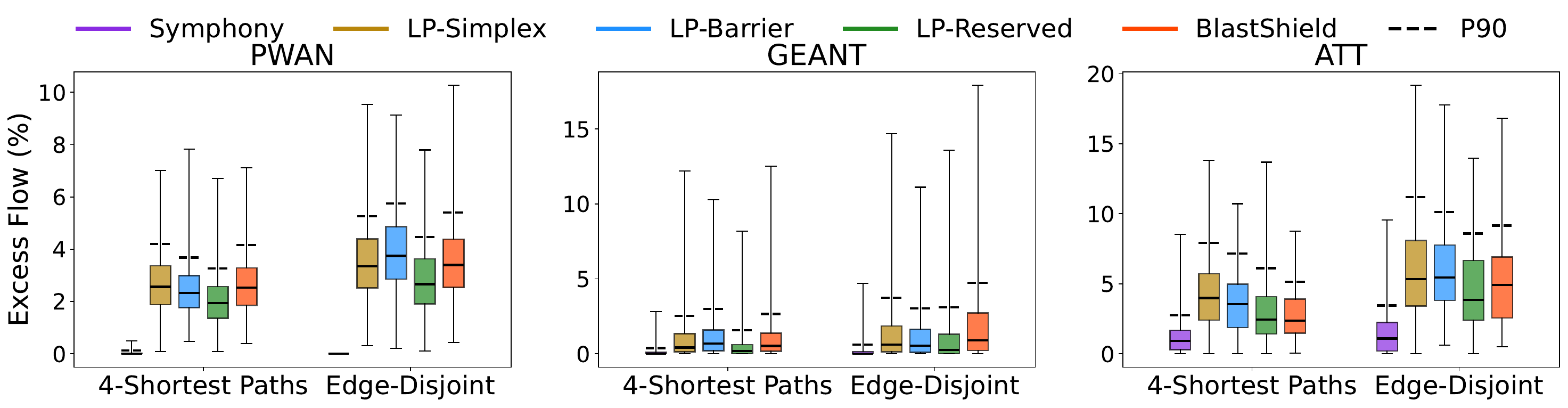}
   \vspace{2mm}
    \caption{The percentage of flow that exceeds link capacities for the maximum throughput objective.}
     \label{fig:divergence-maxflow}
\end{figure*}

\begin{figure*}[h]
      \centering   
    \includegraphics[width=0.75\textwidth]{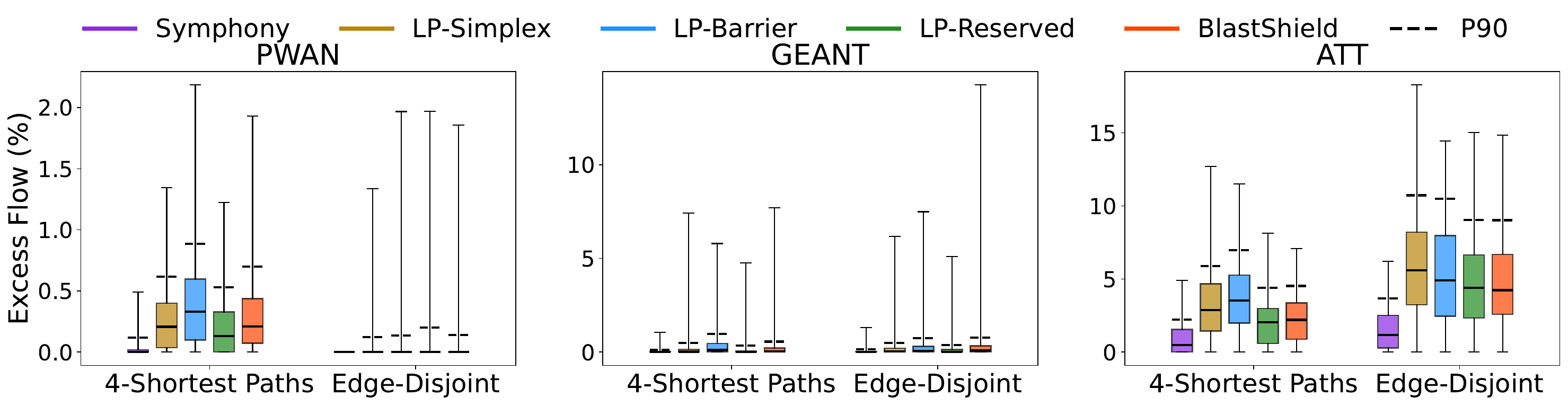}
   \vspace{2mm}
    \caption{The percentage of flow that exceeds link capacities for the maximum concurrent flow objective.}
     \label{fig:divergence-mcf}
\end{figure*}

\begin{figure*}[h]
  \centering   
  \includegraphics[width=0.75\textwidth]{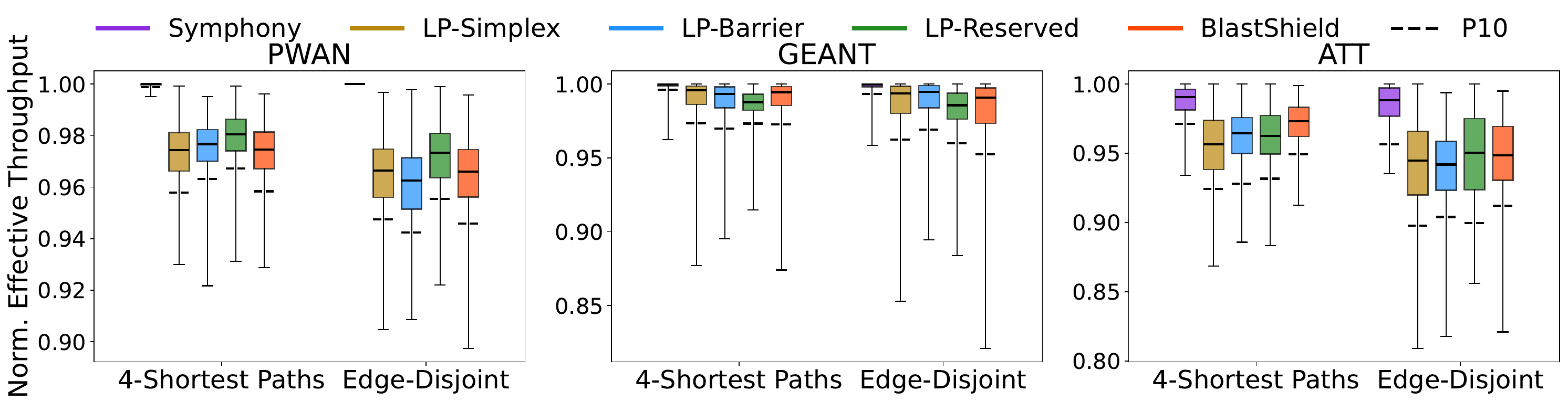}
 \vspace{2mm}
  \caption{The effective throughput for the maximum throughput objective, normalized by the \oracle{'s} throughput.}
   \label{fig:effective-flow}
\end{figure*}

We randomly sample 1000 demand matrices for \meta{} and GEANT, and 
100 for ATT (due to limited available data~\cite{teavar}). Since we simulate demands for KDL, we generate one base demand matrix and use it for 50 iterations with the top 20\% of demands, as demand matrices are often sparse according to prior work~\cite{onewan} and based on our analysis of \meta{}. In each iteration, we generate $k$ perturbed demand matrices, one per slice, from the iteration's sampled demand matrix. The perturbed demand is generated by sampling from the empirical distribution of demand deviations between slice controllers collected from the decentralized production WAN of a large commercial cloud provider. We then run \sysnamete{} and baselines for all three TE objectives in each slice of the WAN using both 4-shortest paths and edge-disjoint paths. We use $\lambda=$1 for MT and $\lambda=$1e-4 for MCF and MMLU. We find that modifying $\lambda$ within several orders of magnitude has no effect on performance (see \S\ref{subsec:appendix-reg-param}). For KDL, we also evaluate the runtime-optimized version, \sysname{-Opt} (see \S\ref{sec:scaling}).

\begin{table}[h!]
\centering
\small
\begin{tabular}{lccc}
\textbf{Topology} & \textbf{\# Nodes} & \textbf{\# Links} & \textbf{Demands} \\
\midrule
GEANT & 23 & 37 & 4.5 months \\
ATT   & 25 & 112  & YATES \\
PWAN  & O(50) & O(400) & O(6) months \\
KDL   & 754 & 895 & Gravity model \\
\end{tabular}
\vspace{3mm}
\caption{WAN Topologies Evaluated}
\label{tab:topologies}
\end{table}

\begin{figure*}[h]
  \centering   
  \includegraphics[width=0.75\textwidth]{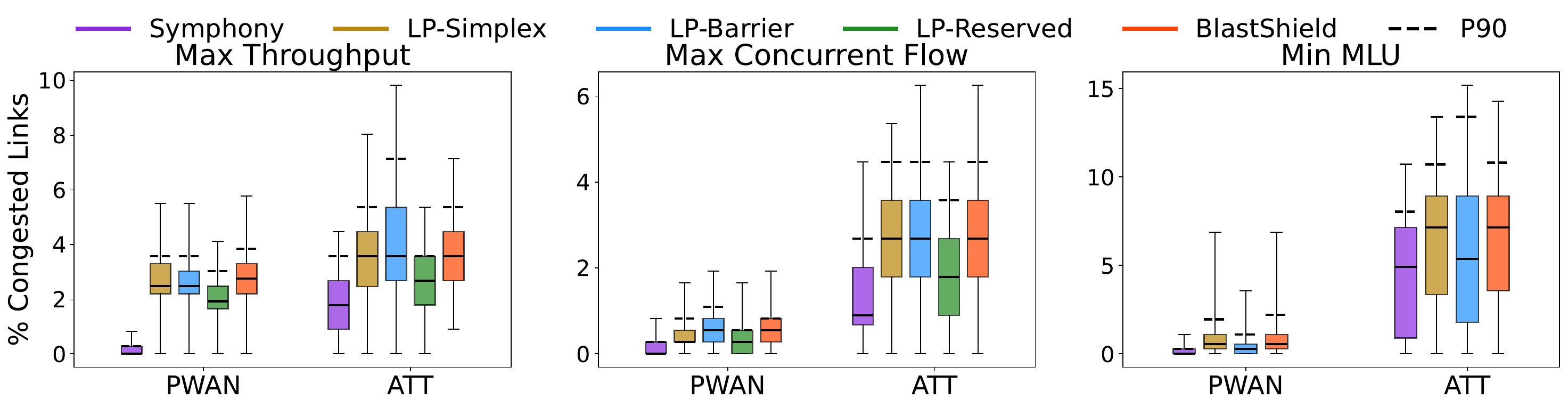}
 \vspace{2mm}
  \caption{The percentage of links that are congested in each iteration, using 4-shortest paths.}
   \label{fig:congestion-incidence}
\end{figure*}
\begin{figure*}[h]
    \centering   
    \includegraphics[width=0.8\textwidth]{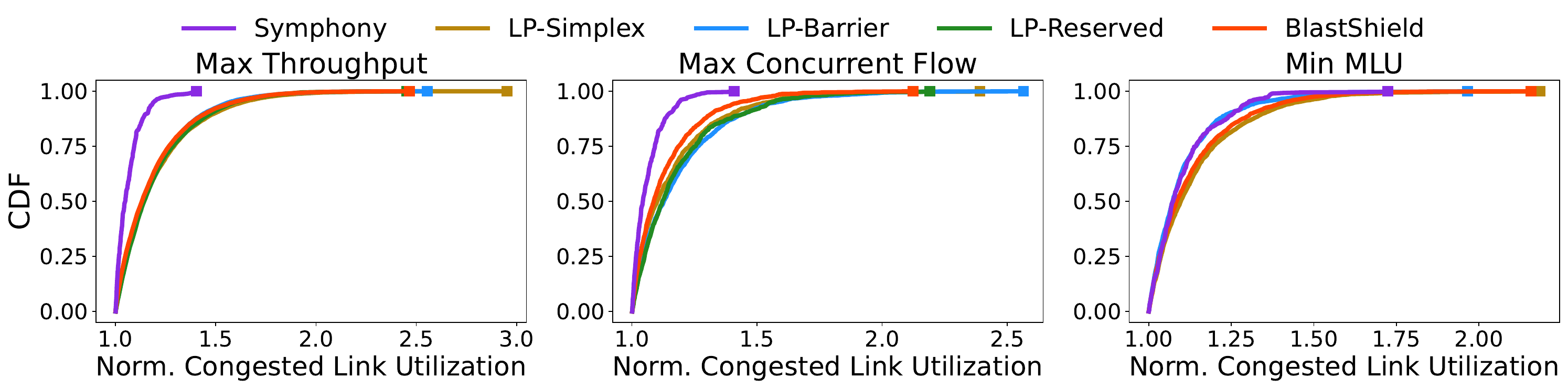}
   \vspace{2mm}
    \caption{CDFs of link utilizations for congested links, across all iterations, for PWAN with 4-shortest paths.}
     \label{fig:congestion-magnitude}
\end{figure*}

\myparab{Baselines.}We mainly compare \sysnamete{} to classical TE controllers used in decentralized WANs.  While deep learning has yet to see production deployment for TE, we run separate experiments to compare with \dote~\cite{dote}, a deep-learning baseline. We use source routing for all controllers besides \blastshield, and solve with 
Gurobi~\cite{gurobi} for all classical schemes.
\begin{compactitem}[leftmargin=*]
\item \textbf{\simplex:} Solves the standard TE LP problem using the Simplex optimization algorithm.
\item \textbf{\barrier:} Solves the standard TE LP problem using the Barrier (\ie interior point) optimization method.
\item \textbf{\scratch:} Reserves 5\% capacity on all links to provide ``slack'' for divergence, thereby solving the standard TE LP problem but with capacity constraints $\leq 0.95 \cdot c_{e}$.
\item \textbf{\blastshield \cite{blastshield}:} Solves the standard TE LP problem, but uses BlastShield's slice routing instead of source routing. Due to \blastshield{'s} slice routing protocol that prevents path weights from being zero, all weights are $\geq 0.001$. 
\item \textbf{\dote \cite{dote}:} A DNN, trained on historical demands, that is more robust to demand fluctuations. Due to its entirely different pipeline, we conduct a separate evaluation in \S\ref{sec:eval-dote}.
\end{compactitem}

For all objectives, we implement an \oracle controller that \emph{cannot experience divergence}. \oracle is a centralized controller that knows the exact 
predicted demands in each slice and computes a flow allocation accordingly. Thus, \oracle 
can sustain the maximum throughput and concurrent flow \emph{without any congestion} and provides 
a baseline level of congestion for evaluating the MMLU objective (see \S\ref{subsec:mlu-formulation} for more details). We also considered (semi-)oblivious approaches such as COPE~\cite{cope, obliviouste}, but these have only been formulated for the MMLU objective, do not scale, and are outperformed by much more recent baselines that we directly compare to. 

\begin{figure}[t]
  \centering   
 \includegraphics[width=0.87\linewidth]{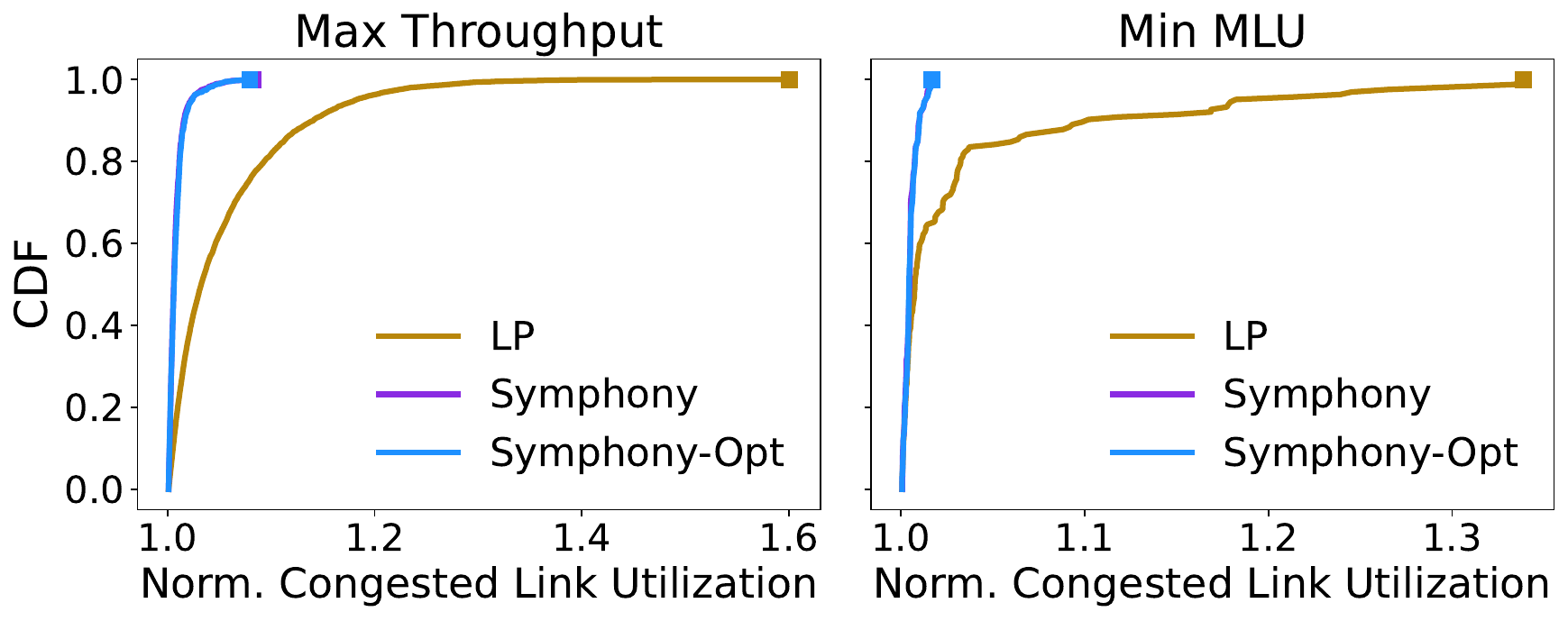}
  \caption{Evaluation on KDL using 4-shortest paths.}
   \label{fig:kdl-results}
\end{figure}

\myparab{Key metrics.} We compare \sysname{} and baselines across the following metrics for decentralized WAN performance. \textbf{Effective throughput:} Throughput from computed flow allocations, subtracted by the sum of link over-allocations, and then normalized by the throughput of \oracle. This is often called \emph{congestion-free throughput}. \textbf{Excess flow:} The percent of total flow that exceeds link capacities and thereby causes congestion. \textbf{Congested links:} We measure both the \emph{percent} of links that are congested and the \emph{utilizations} of congested links (normalized by the \oracle{'s} MLU). A link is congested if its utilization exceeds $1$ for MT and MCF and if it exceeds the \oracle{'s} MLU for MMLU, since even the \oracle{} can be congested due to the MMLU formulation's lack of a capacity constraint.

\subsection{Lowers congestion for all objectives}
\sysname{} substantially lowers divergence-induced congestion in decentralized TE across all objectives, topologies, and path selection strategies. In all boxplots, whiskers are the minimum and maximum values observed over all iterations.

\myparab{Excess flow and effective throughput.} While congestion directly captures performance for the MMLU objective, which ensures all demands are satisfied, effective throughput and excess flow capture how well congestion is managed for MT and MCF, where the amount of flow sent can differ between baselines (\ie to avoid comparisons where sending much less flow reduces congestion). Figures~\ref{fig:divergence-maxflow} and~\ref{fig:divergence-mcf} show that across all path selection strategies and topologies, \sysname{} has the minimum percent of flow that exceeds link capacities for both MT and MCF. \sysname{} sustains \emph{no excess flow} on average and $14\times$ less than \blastshield in the worst case. Figure~\ref{fig:effective-flow} shows that \sysname{} sustains the highest effective throughput over all baselines: in \meta, 99.9\% (average) and $\geq$99.5\% (worst case) of the divergence-free \oracle{}. \sysname{} performs especially well in the worst case, as its penalty scales quadratically with link utilization. For MT and MCF, \sysname{} has the lowest congestion while sustaining high throughput.
\begin{figure}[t]
  \centering
  \begin{subfigure}[b]{0.48\linewidth}
    \includegraphics[width=\textwidth]{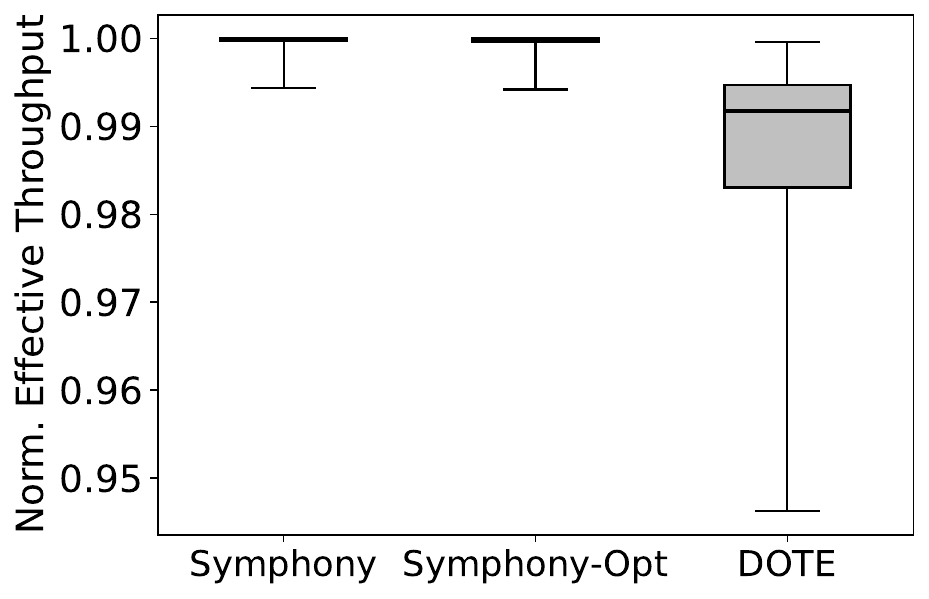}
   \caption{\small{MT: random.}}
   \label{fig:max-flow-dote}
\end{subfigure}
  \begin{subfigure}[b]{0.48\linewidth}
      \includegraphics[width=\linewidth]{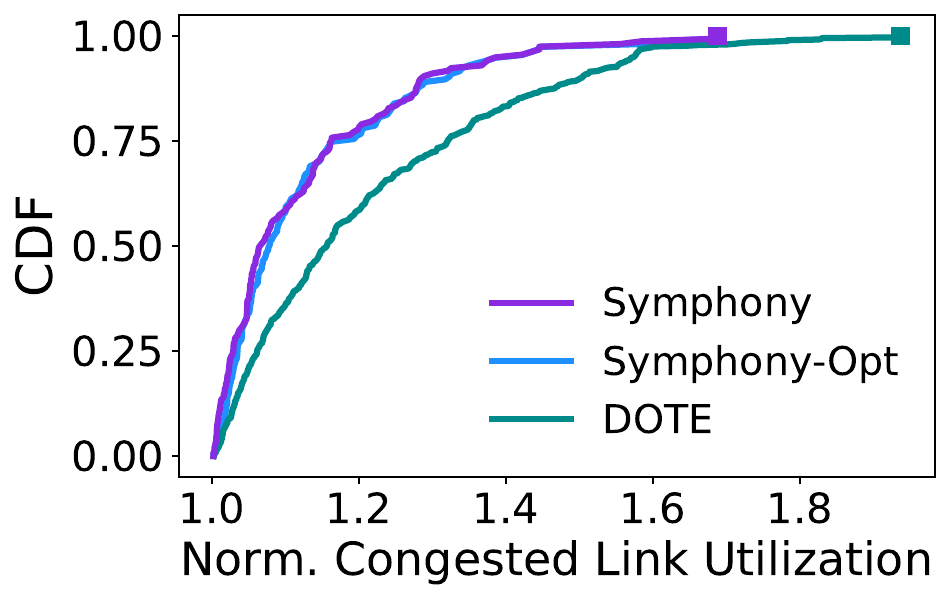}
     \caption{\small{MMLU: worst-case.}}
     \label{fig:min-mlu-dote-worst}
    \end{subfigure}
 \vspace{3mm}
  \caption{Comparison with DOTE using 4-shortest paths.}
   \label{fig:dote}
   \vspace{2mm}
\end{figure}

\begin{figure*}[t]
    \centering
    \begin{subfigure}[b]{0.24\textwidth}
      \includegraphics[width=\textwidth]{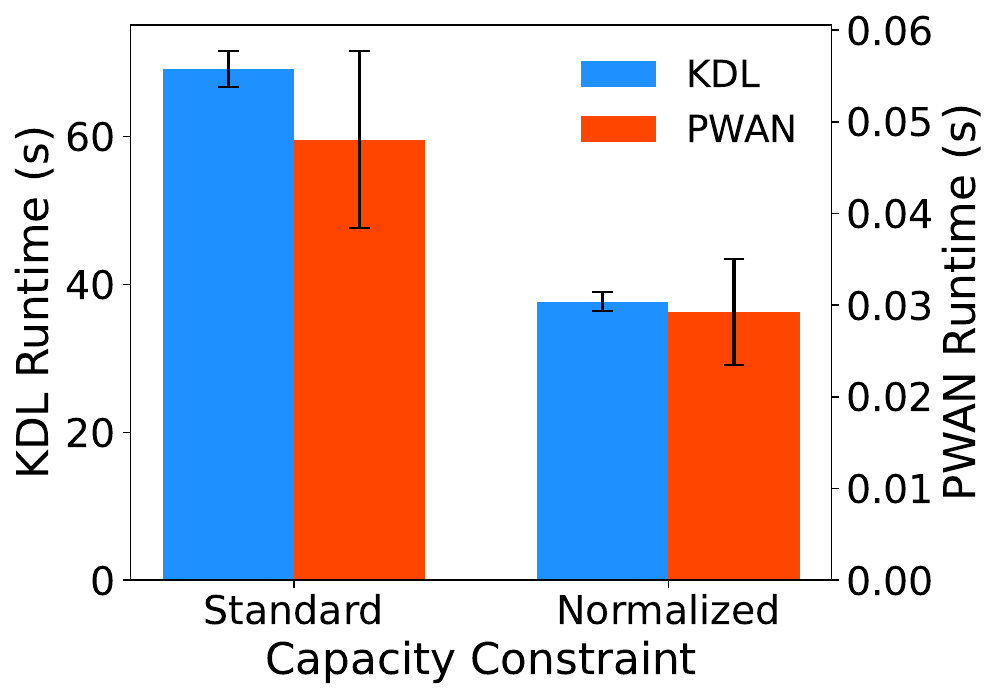}
     \caption{\small{Max throughput runtimes.}}
     \label{fig:max-flow-runtime}
  \end{subfigure}
    \begin{subfigure}[b]{0.34\textwidth}
      \includegraphics[width=\textwidth]{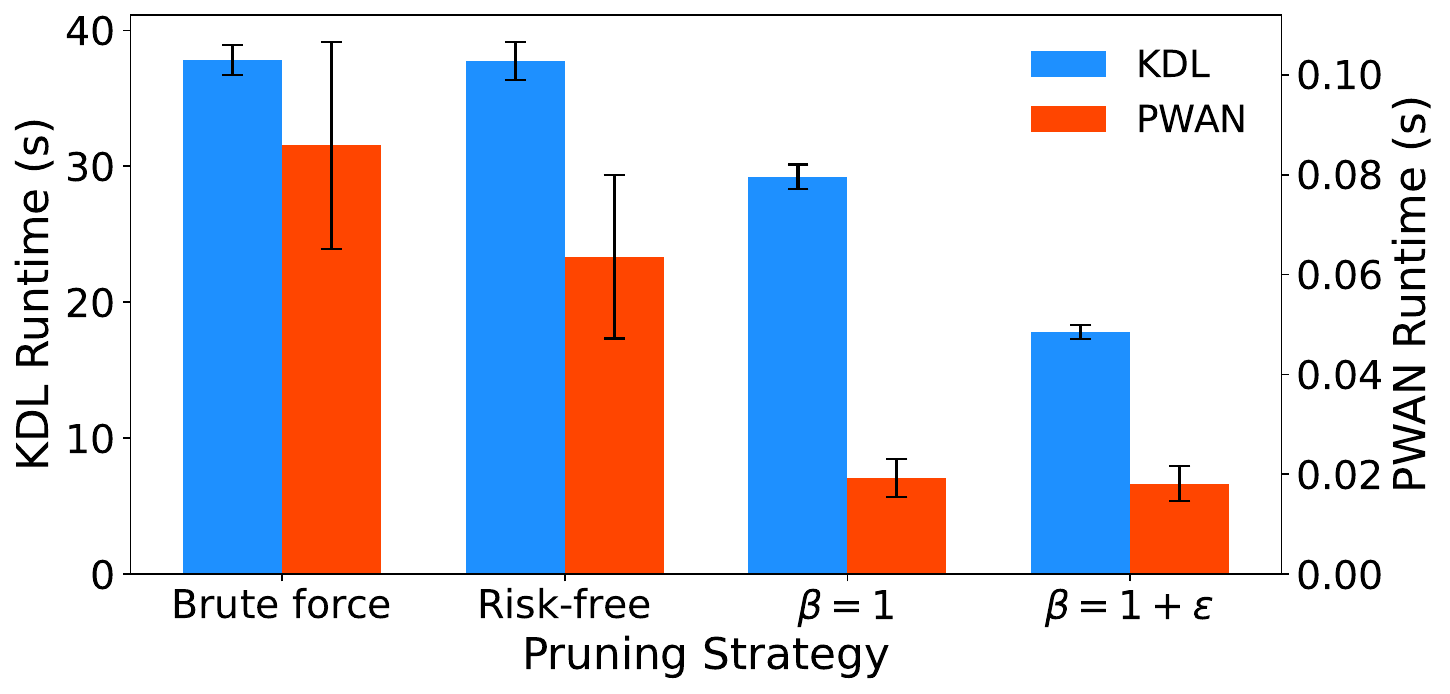}
     \caption{\small{Min MLU runtimes.}}
     \label{fig:min-mlu-runtime}
    \end{subfigure}
    \begin{subfigure}[b]{0.34\textwidth}
        \includegraphics[width=\textwidth]{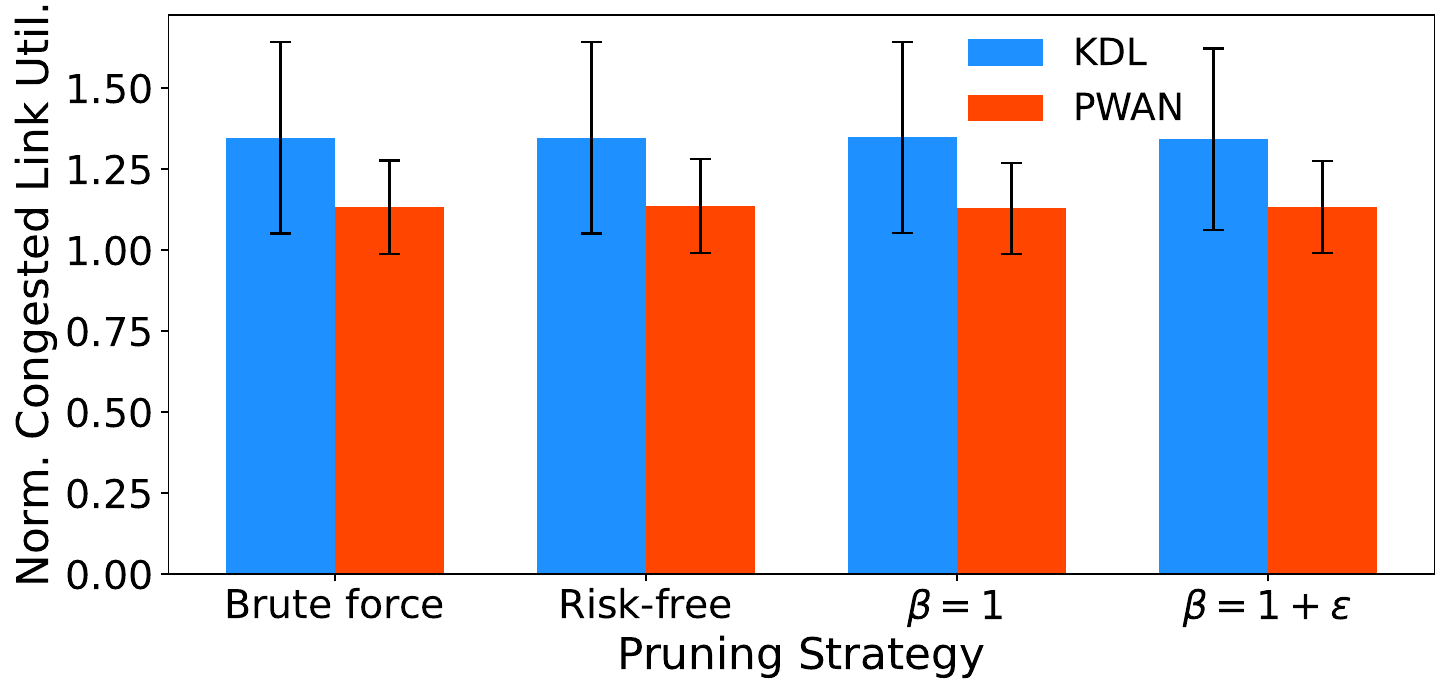}
       \caption{\small{Min MLU solution quality.}}
       \label{fig:min-mlu-optimality}
      \end{subfigure}
   \vspace{3mm}
    \caption{Ablation studies for \sysname{'s} runtime optimizations, based on 100 iterations (700 TE computations, due to 7 slices) for \meta{} and 1 iteration  (25 TE computations, due to 25 slices) for KDL. Bar heights are the means and error bars are standard deviations. In \ref{fig:min-mlu-runtime} and \ref{fig:min-mlu-optimality}, pruning strategies are applied additively from left to right, and $\epsilon=0.1$ for \meta{} and $\epsilon=1$ for KDL.}
     \label{fig:runtimes}
     \vspace{2mm}
\end{figure*}

\myparab{Congested links.} The key congestion metrics for all three objectives are the percent of links that are congested and the distribution of their utilizations. Figure~\ref{fig:congestion-incidence} shows that across topologies and objectives, \sysname{} congests the fewest links due to divergence. In \meta{}, \sysname{} congests $7\times$ and $5.8\times$ fewer links than \blastshield{} in worst-case scenarios for the MT and MMLU objectives, respectively. Figure~\ref{fig:congestion-magnitude} 
shows CDFs of utilizations of congested links (normalized by the \oracle{'s} MLU). Across all objectives, \sysname{} reduces congestion, 
especially in worst-case scenarios, and generalizes to all topologies evaluated (see \S\ref{sec:appendix-congestion}). 
For MT in \meta{}, \sysname{} ensures 90\% of congested links are oversubscribed by <15\% while reducing the maximum oversubscription by 72.4\% and 79.3\% 
compared to \blastshield and \simplex, respectively. For KDL, we only compare \sysname{} and \sysname{-Opt} to \simplex due to long solver times for baselines and find that both versions consistently and near-identically
reduce link oversubscription (Figure~\ref{fig:kdl-results}).

\subsection{Comparison with learning-based TE}
\label{sec:eval-dote}
We compare \sysname{} with \dote~\cite{dote} on \meta{} for the MT and MMLU objectives, with a different set of experiments.
We choose 300 consecutive demand matrices in the first 1/4th of timestamps, perturb these to create a training 
dataset for each slice, and for each slice, train \dote for 50 epochs for each objective (since a separate model 
would run TE for each slice in a decentralized WAN). We then run TE on 100 demand matrices, matching \dote{'s} 
recommended 75/25 train/test split. In the first experiment, we randomly sample these demand matrices from later 
timestamps to test average behavior while in another, we choose later demand matrices with the greatest Euclidean 
distance from the demands used in training to compare worst-case performance. Figure~\ref{fig:max-flow-dote} shows 
that, even in the average case, \sysname{} outperforms \dote{} in maximizing congestion-free throughput for the MT objective. 
For MMLU, \dote performs comparably with \sysname{} in average scenarios (see \S\ref{subsec:dote-appendix}), but  
underperforms in worst-case scenarios (Figure~\ref{fig:min-mlu-dote-worst}), where the model receives inputs that deviate 
from its training data. Unlike \dote{}, \sysname{} is \emph{oblivious} to \emph{historical} demands, enabling it to generalize well in out-of-distribution settings, 
and does not require retraining on large solver-labeled datasets.

\subsection{Scaling optimizations and TE runtime}
\label{sec:eval-runtime}

Figure~\ref{fig:runtimes} shows ablation studies that quantify \sysname{'s} runtime improvements after applying the optimizations in \S\ref{sec:scaling}, for the larger topologies and the MT and MMLU objectives. Figure~\ref{fig:max-flow-runtime} shows how normalizing capacity constraints nearly halves the runtime for both \meta{} and KDL. Figure~\ref{fig:min-mlu-runtime} shows how pruning strategies, from removing divergence-free links to $\beta$-constrained links (for different values of $\beta$), reduce the MMLU runtime by $2.12\times$ in KDL and over $4.74\times$ in \meta{}, while Figure~\ref{fig:min-mlu-optimality} shows that these runtime optimizations do not impact \sysname{'s} performance in reducing divergence-induced congestion. After optimizations, \sysname{} on average solves \emph{as quickly as} an LP on KDL: $37.7 s$ (\sysname{}) vs. $36.7 s$  (LP) for MT, and $17.8 s$ vs. $15.8 s$ for MMLU.
\section{Partitioning the WAN for resilience}
\label{sec:slicing}
While \sysname{} provably reduces divergence-induced congestion regardless of the slicing strategy, slicing itself is central to decentralized TE.
Slicing not only determines the extent of fault isolation, but also shapes how traffic 
allocations diverge across the network. Thus, to fully understand the role of \sysname in decentralized TE, we ask: 
\emph{how does the slicing structure influence both divergence and fault isolation?} 

To answer this, we design and evaluate a randomized slicing algorithm that generates hundreds of feasible slicing 
configurations within minutes and achieves a blast radius near the theoretical minimum, outperforming state-of-the-art decentralized WAN architectures. Our analysis reveals a fundamental tension: slicing strategies that reduce divergence often concentrate traffic within a single slice, amplifying the impact of faults. Crucially, because \sysname{} stabilizes allocations regardless of the slicing configuration, slicing can focus on its primary role of limiting the blast radius. We validate this by evaluating \sysnamete{} across slicing structures produced by both our algorithm and naive approaches, showing that it consistently outperforms all baselines in mitigating divergence-induced congestion \emph{regardless of the slicing configuration}.

\subsection{Limiting the blast radius of faults}
Decentralized WANs partition the WAN into slices of nodes in order to limit the \emph{blast radius} of controller faults. 

\myparab{Existing graph partitioning algorithms are insufficient.} Graph partitioning is notoriously NP-hard, and existing algorithms~\cite{kernighan1970efficient, newman2006modularity, blondel2008fast, hagen1992new} are insufficient for decentralized WANs. First, these algorithms often aim to minimize the number of edges cut by a partition or maximize the density of edges within a partition---properties with no clear relationship to our goal of fault isolation in the control plane. Further, few algorithms handle $k$-way partitioning (as opposed to bipartitioning), a necessity for WAN slicing, or ensure partitions are traffic-balanced, which is critical for fault tolerance.

\myparab{Shielding resources.} Imbalanced slices, in terms of nodes per slice, are unfavorable for fault isolation since controller faults in larger slices have a disproportionate impact on datacenters.

\myparab{Source routing defines the blast radius.} In \blastshield, the blast radius of a slice is defined as the percent of traffic impacted by a controller fault. However, slice routing complicates the definition of the blast radius, as traffic can be impacted by controller faults for \emph{any} slices it \emph{transits} through. Thus, the amount of traffic impacted by a controller fault depends on the path through the network, which is difficult to reason about and changes with every TE iteration. Designing WAN slices to minimize the expected blast radius using \blastshield{'s} routing strategy would be combinatorially explosive. In contrast, source routing clearly defines the blast radius as the maximum percent of traffic that \emph{originates} from any slice, capturing the worst-case traffic impact of a controller fault.
\begin{definition}
The blast radius, $r$, of a slicing configuration with $k$ slices $S_{j}$, where $j$=$1$,...,$k$ and $\bigcup_{j=1}^{k} S_j = V$, is given by
\begin{equation}
    r = \max_{1 \leq j \leq k} \sum_{\text{src}(i) \in S_{j}} d_{i}
\end{equation}
\end{definition}

\myparab{Shielding traffic.} Based on the definition of $r$, each WAN slice should contain a roughly equal amount of \emph{egress} demand. While demands fluctuate, their relative magnitudes are consistent and can be approximated with the average past demand. Because source routing assigns flows to the controller of their source node, this approach aims to minimize $r$ in expectation.
\subsection{A randomized algorithm for fault-tolerant slicing}
We propose a novel randomized algorithm to generate candidate slicing configurations with low blast radius by design. (A \emph{slicing configuration} is a unique set of slices that constitutes a solution.) This approach provides the operator with flexibility to select a slicing configuration among candidates that best suits operational needs. We first describe the algorithm (with more detail in Algorithm~\ref{alg:balanced_partitions} of \S\ref{appendix:ilp}) and then discuss how slicing impacts divergence in \S\ref{subsec:slicing-divergence}.

\myparab{Node weights.} First, we define a node $v$'s weight, $\phi_{v}$, as the sum of the average or maximum demand value for all flows originating at $v$. This estimates how much traffic originates from $v$.

\myparab{Inputs.} The operator provides node weights $\phi_{v}$, the desired number of slices $k$, target slice sizes $s_{j}$, target per-slice total weight $T = (\sum_{v \in V} \phi_{v})/k$, and a tolerance $\epsilon \in [0,1]$ that specifies the deviation from $T$ that still allows the solution to be deemed feasible. For example, $\epsilon = 0.2$ means the operator requires the total weight of a slice to be within 20\% of $T$. Higher $\epsilon$ values relax the problem but reduce fault-tolerance guarantees.

\myparab{Initializing slices with elephant sources.} Our algorithm aims to assign nodes to slices to balance the total weight, $\sum_{v \in s_{j}} \phi_{v}$, per slice; this avoids concentrating too much traffic in any one slice. However, the distribution of $\phi_{v}$ is far from uniform, leading to a few \emph{elephant sources} that send a vast majority of the traffic (see \S\ref{appendix:elephant}). We can leverage this to extensively prune the search space, by initializing each slice with an elephant source. If $v_{1}$ is an elephant source assigned to slice $S_{1}$, assigning another elephant source $v_{2}$ also to $S_{1}$ would \emph{immediately} violate the constraint that all slices should have total weight within $T(1\pm\epsilon)$.

\myparab{$k$-way BFS.} With each slice initialized, the algorithm aims to construct slices that are (1) connected, (2) roughly equal in the number of nodes, and (3) roughly balanced in weight.

In each round, the algorithm iterates through each slice and assigns a node to the slice. Each slice maintains a list of \emph{candidates}, which are neighbors of any nodes that have already been included in the slice. When a node is selected by one slice, it is removed from other slices' candidates lists.


The algorithm terminates successfully when all slices have met size constraints and have total weight within $(1 \pm \epsilon)T$, and unsuccessfully if not (or if candidates lists for any slices become empty). If no solution is found in one run, the algorithm re-initializes. Like many approximation algorithms for graph partitioning, the algorithm does not guarantee a solution will be found in \emph{every} iteration, but generates many feasible solutions in practice due to its efficient runtime (see \S\ref{subsec:slicing-divergence}) and randomization. Further, unlike many approximations, such as LP relaxations, our algorithm guarantees that solutions discovered truly meet fault-tolerance objectives.

\subsection{Evaluation and impact on divergence}
\label{subsec:slicing-divergence}
We generate 100 slicing configurations for \meta{}/GEANT and 50 for ATT with our algorithm, and select the most weight-balanced configuration for the evaluations in \S\ref{sec:eval-divergence}. We generate $k\approx\sqrt{n}$ slices with roughly $\sqrt{n}$ nodes each. For \meta{}, we run experiments to assess the impact of slicing on divergence.

\begin{figure}[b]
    \centering
    \includegraphics[width=0.6\linewidth]{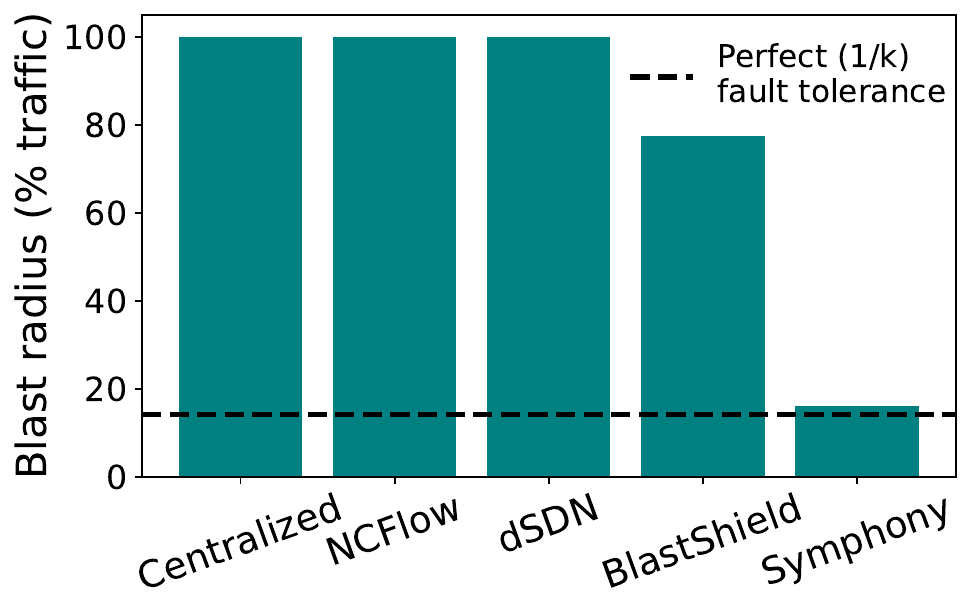}
    \caption{Blast radius for different WAN architectures.}
    \label{fig:blast-radius}
\end{figure}

\begin{figure*}[t]
    \centering
    \begin{subfigure}[b]{0.3\textwidth}
      \includegraphics[width=\textwidth]{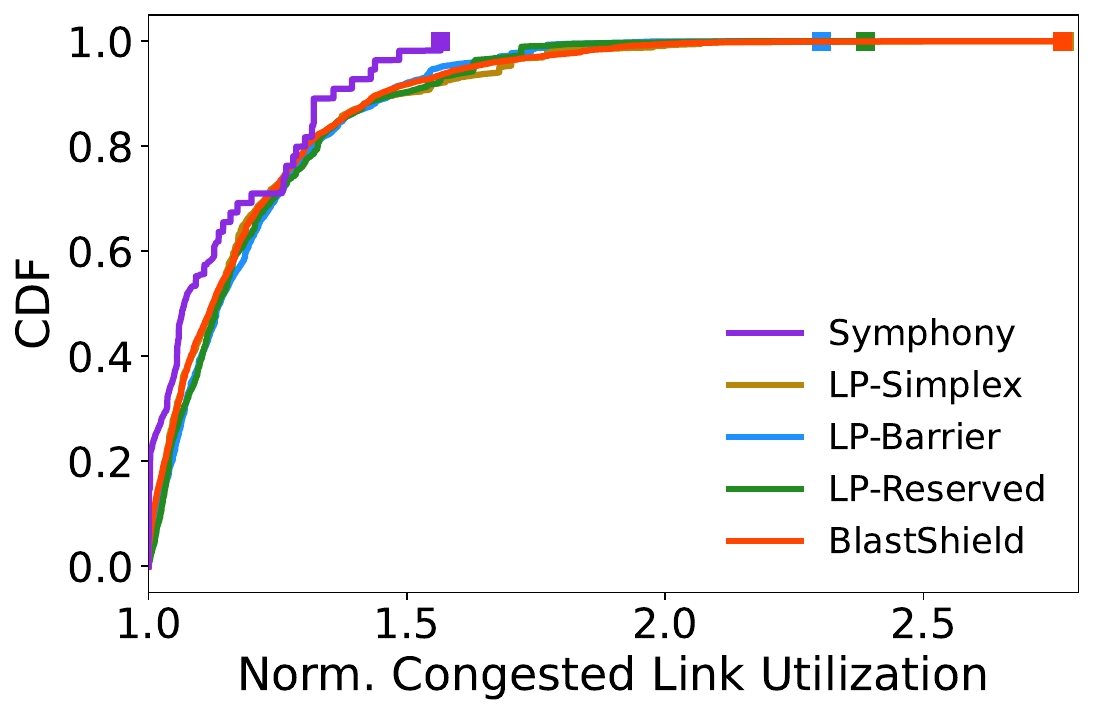}
     \caption{\small{Balanced-slicing congestion.}}
     \label{fig:balanced-slice-cdfs}
  \end{subfigure}
    \begin{subfigure}[b]{0.325\textwidth}
      \includegraphics[width=\textwidth]{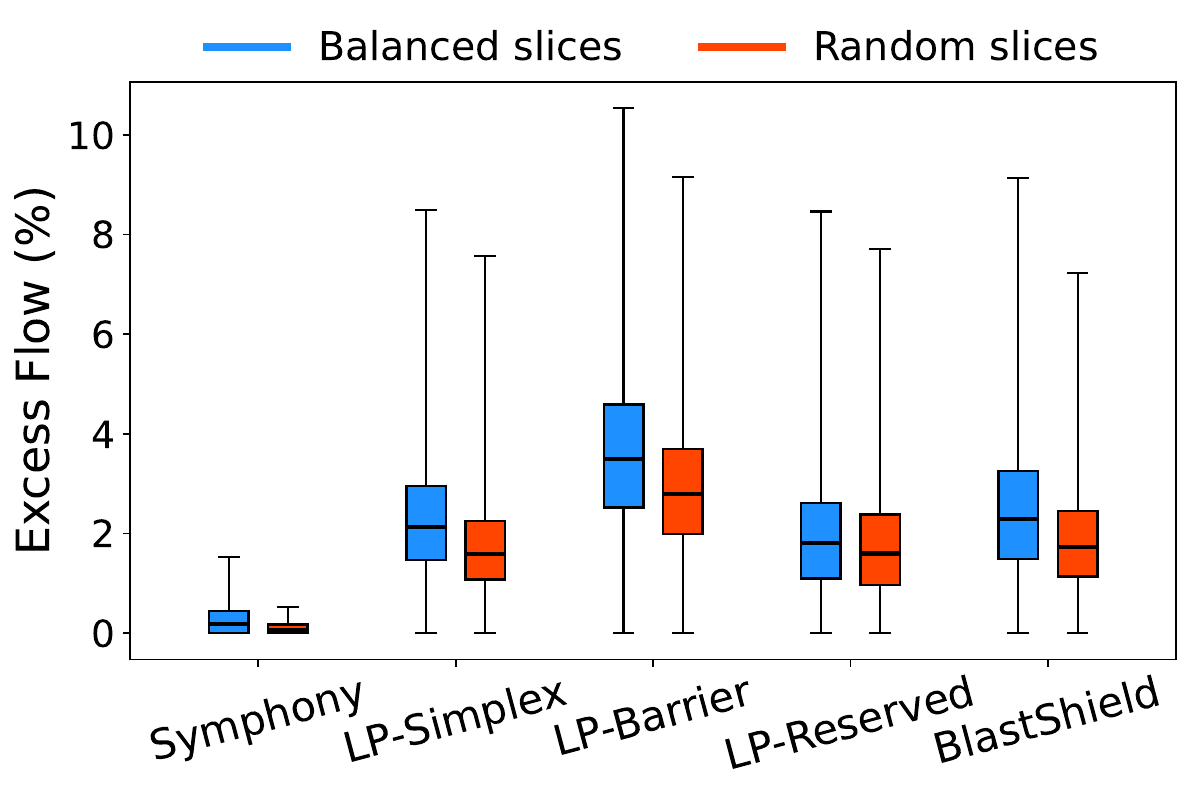}
     \caption{\small{Balanced vs.~random slicing.}}
     \label{fig:balanced-vs-random}
    \end{subfigure}
    \begin{subfigure}[b]{0.32\textwidth}
        \includegraphics[width=\textwidth]{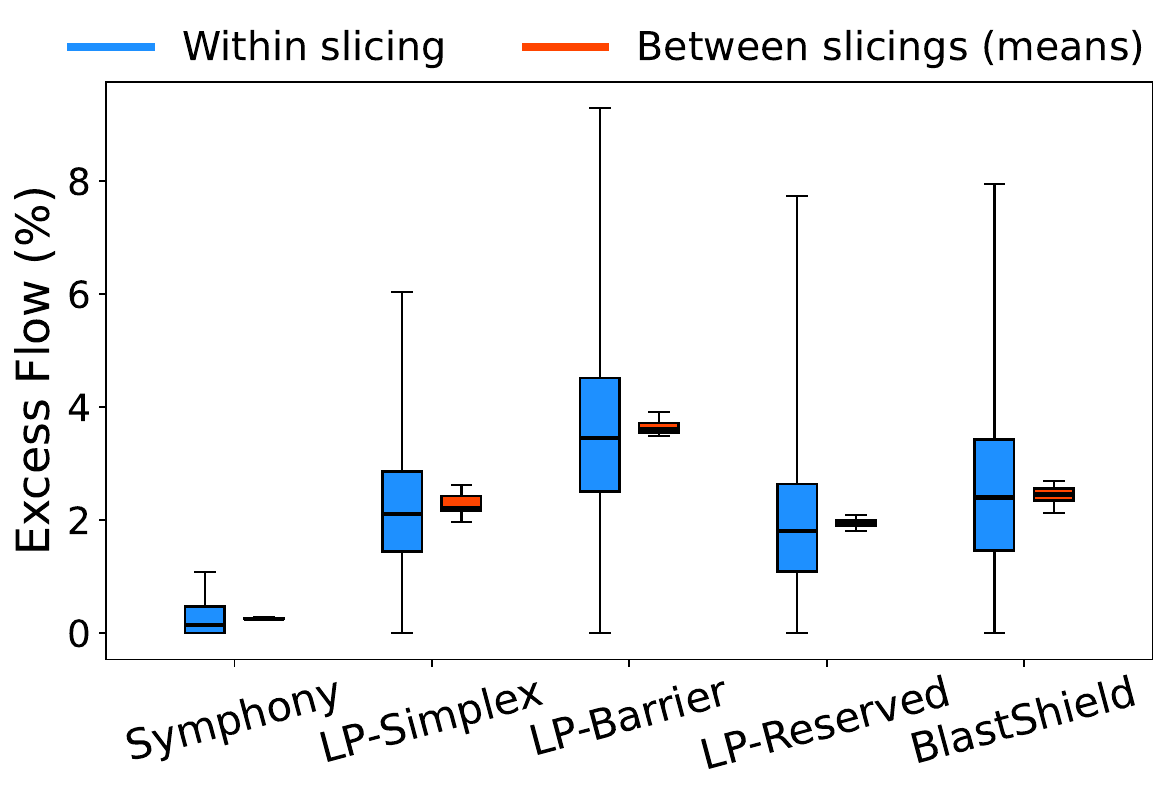}
       \caption{\small{Slicing impact on divergence.}}
       \label{fig:intra-slicing}
      \end{subfigure}
   \vspace{3mm}
    \caption{\ref{fig:balanced-slice-cdfs} plots congested link utilizations over all permutations, for all slicing configurations generated by our algorithm. \ref{fig:balanced-vs-random} compares congestion between weight-balanced slices (via our algorithm) and randomly generated slices only meeting size and connectivity constraints.~\ref{fig:intra-slicing} compares the distribution in excess flow over permutations of perturbed demand matrices for a single slicing configuration (\emph{within}) with the distribution of average excess flow (after averaging all permutations) across configurations (\emph{between}). In all boxplots, whiskers denote the maximum and minimum values observed.}
     \label{fig:slicing}
     \vspace{2mm}
\end{figure*}

\myparab{Methodology.} We specify $k=7$ slices (4 with 7 nodes and 3 with 8 nodes), $\epsilon=0.2$ tolerance, and use mean demands over 6 months to calculate node weights for \meta{}. We generate $k$ slice demands by perturbing the  demand matrix according to the empirical demand deviation distribution. We fix these demand matrices for all experiments to isolate the impact of the slicing scheme. We then generate 100 slicing configurations using Algorithm~\ref{alg:balanced_partitions}. For each candidate slicing configuration, we evaluate the divergence-induced congestion on the MT objective using the precomputed per-slice demand matrices. To isolate the impact of the slicing configuration from the \emph{exact matching} of perturbed demand matrices to slices, we evaluate the divergence across all $k!$ permutations of slice demand matrices. Thus, our experiments involve 100 slicing configurations and 5040 iterations per slicing configuration. 

\myparab{Runtime.} The algorithm generates 100 \emph{unique} slicing candidates meeting all constraints in under 12 minutes on CPU. We note that this algorithm would be run offline infrequently, as operators are not regularly re-slicing their WANs. While scaling to much larger WANs poses a challenge, operators can manually pre-slice their WAN into just a few  ``meta-slices'' and then run the algorithm within each meta-slice.

\myparab{Near-perfect fault tolerance.} Figure~\ref{fig:blast-radius} compares the blast radius of slices generated by our algorithm
with today's WAN architectures. WANs that rely on a central demand predictor, such as centralized controllers in  
NCFlow~\cite{ncflow} and dSDN~\cite{dsdn}, expose \emph{all} traffic to faults. \blastshield{} 
decentralizes demand prediction, but traffic is impacted by controller faults in any slices it transits through. 
In contrast, our algorithm enforces fault tolerance, and finds a solution that limits the blast radius to 15\% of traffic (79\% lower than \blastshield), 
nearly indistinguishable from a theoretically perfect blast radius of 14\% (1/7) if every slice had exactly 
equal weight (which may be infeasible in practice due to topology and node weights).

\myparab{Reduced congestion regardless of slicing.} Figure~\ref{fig:balanced-slice-cdfs} shows---across all slicing configurations generated by our algorithm and all permutations---\sysnamete{} best reduces oversubscription due to divergence in congested links. \sysnamete{} reduces the MLU by 60\% compared to the best baseline.

\myparab{Balanced slices increase divergence.} Using the same permutation experiments, we compare weight-balanced slices (from our algorithm) with 100 randomly generated slicing configurations that ignore 
node weights and simply meet size and connectivity constraints (see \S\ref{appendix:rand-slice} for more results). Figure~\ref{fig:balanced-vs-random} shows that regardless of TE method, balanced (\ie low blast-radius) slices  
incur greater divergence than random slices, highlighting a core tradeoff between stability and fault tolerance. Concentrating the amount of flow that 
originates in any one slice allows more traffic to be engineered using the \emph{same source of truth}, thereby reducing divergence, while also exposing more traffic 
to disruption from a controller fault. As \sysnamete{} reduces divergence across all slicing configurations, whether generated randomly or with our algorithm, we focus solely on 
fault tolerance for slicing.

\myparab{Divergence is driven by noisy demands, not slicing.} We compare, via the permutation experiments, whether exact demand values (which matter to TE) or slicing configuration contributes more to divergence. Figure~\ref{fig:intra-slicing} shows that the variation in excess flow across different permutations for a \emph{single} slicing configuration far exceeds the variation in mean excess flow (over all permutations) across \emph{different} configurations; this suggests divergence is mainly driven by noisy demands, not slicing. Thus, the exact choice of slicing configuration is insignificant for divergence-induced congestion, a problem that \sysnamete{} can handle well.
\section{Related Work}
\label{sec:related}

\myparab{Stability in TE.} WAN TE systems have been optimized for uncertainty~\cite{teavar, ffc, smore-nsdi, miao2025prete}. Early work either operated obliviously~\cite{obliviouste} or optimized for worst-case demands~\cite{cope}, but both face performance and scaling challenges. More recent approaches leverage deep learning to enhance TE stability~\cite{dote, figret, redte, alqiam2025hattrick}, but they require significant changes to existing TE infrastructure and depend on extensive training data for robustness. In contrast, \sysname{} integrates seamlessly into solver-based systems, requires no training, and guarantees stability even under distribution shifts.

\myparab{Decentralized TE architectures.}
Our work builds on a decentralized TE architecture deployed in a large cloud provider's WAN~\cite{blastshield}. Alternative architectures have been proposed~\cite{dsdn,meta-ebb}. For example, another design defines slices as subgraphs containing all nodes but only a subset of links~\cite{meta-ebb,meta-backbone}. 
As this system is susceptible to demand perturbations and controller divergence, our techniques still apply. In dSDN~\cite{dsdn}, each router uses a central demand source for flow allocations. While this eliminates divergence, it introduces a global blast radius when faults or errors occur in controller input services.

\myparab{Scaling TE.}
Recent work has developed techniques to reduce the size of the TE problem on large networks~\cite{ncflow,demand-pinning,soroush}. 
NCFlow~\cite{ncflow, pop-sosp} creates clusters of nodes and solves a smaller LP for each cluster. It then recombines solutions 
from the smaller LPs to compute global flow allocations. However, recombining LPs is error-prone, and the solution can be far from 
optimal~\cite{ncflow}. We note that these methods are unlike the decentralized TE systems deployed in practice~\cite{onewan,meta-ebb}, which 
solve the global TE problem within each slice. Researchers have also proposed alternative ways to scale TE in centralized 
WANs~\cite{teal, miao2024megate}, a problem that is orthogonal to our work.
\section*{Acknowledgements}
\noindent AD is supported by the NSF Graduate Research Fellowship.
\label{endOfBody}

\bibliographystyle{plain}
\bibliography{reference}
\clearpage
\appendix
\appendix
\label{sec:appendix}
\begin{figure*}[t]
    \centering
    \begin{subfigure}[b]{0.3\textwidth}
      \includegraphics[width=\textwidth]{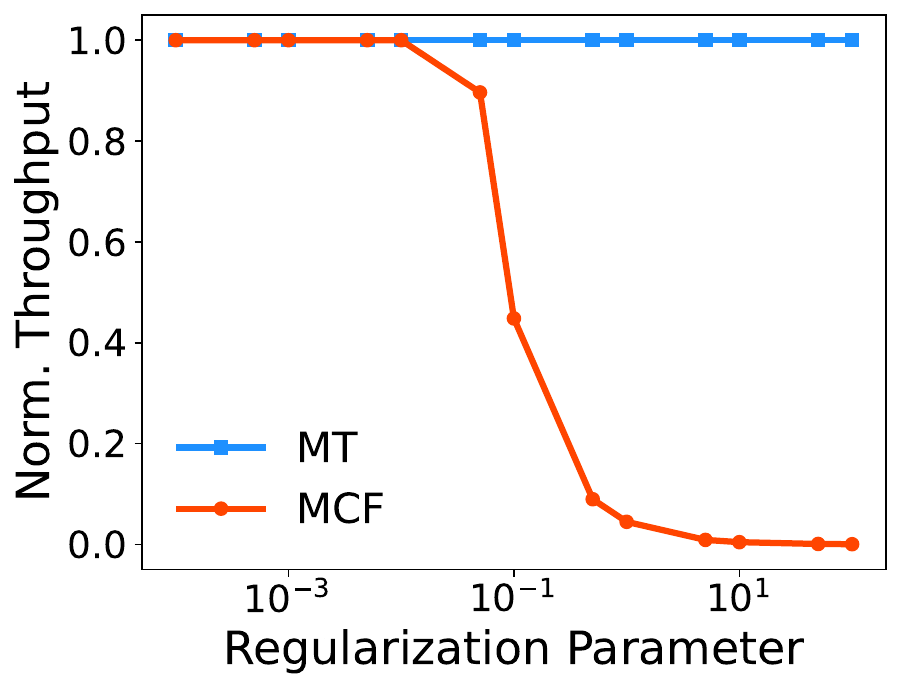}
     \caption{\small{Throughput.}}
     \label{fig:reg-param-throughput}
  \end{subfigure}
    \begin{subfigure}[b]{0.3\textwidth}
      \includegraphics[width=\textwidth]{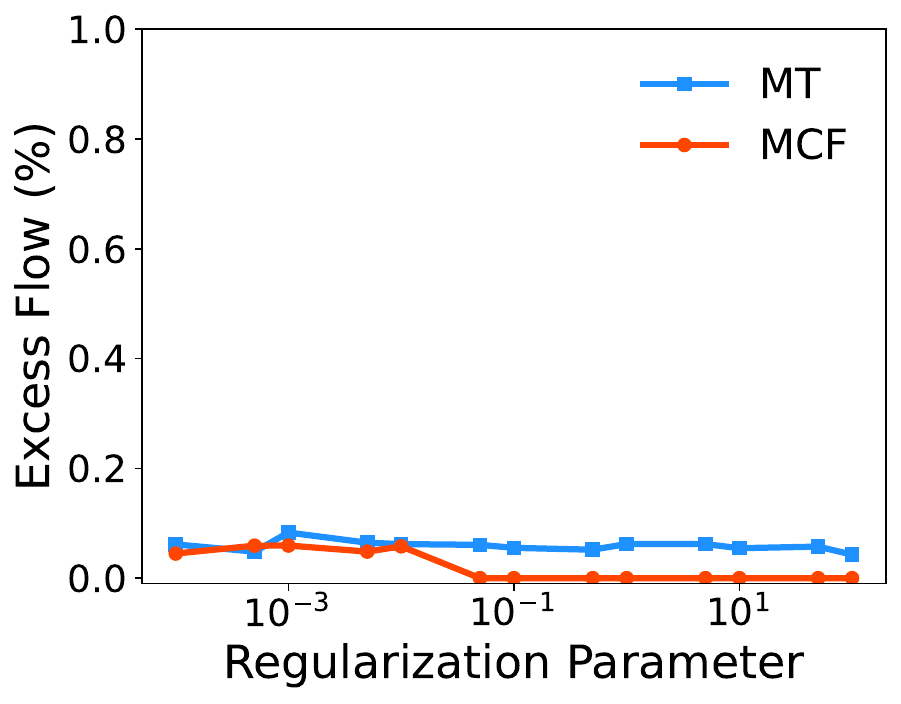}
     \caption{\small{Excess flow.}}
     \label{fig:reg-param-excess-flow}
    \end{subfigure}
    \begin{subfigure}[b]{0.31\textwidth}
        \includegraphics[width=\textwidth]{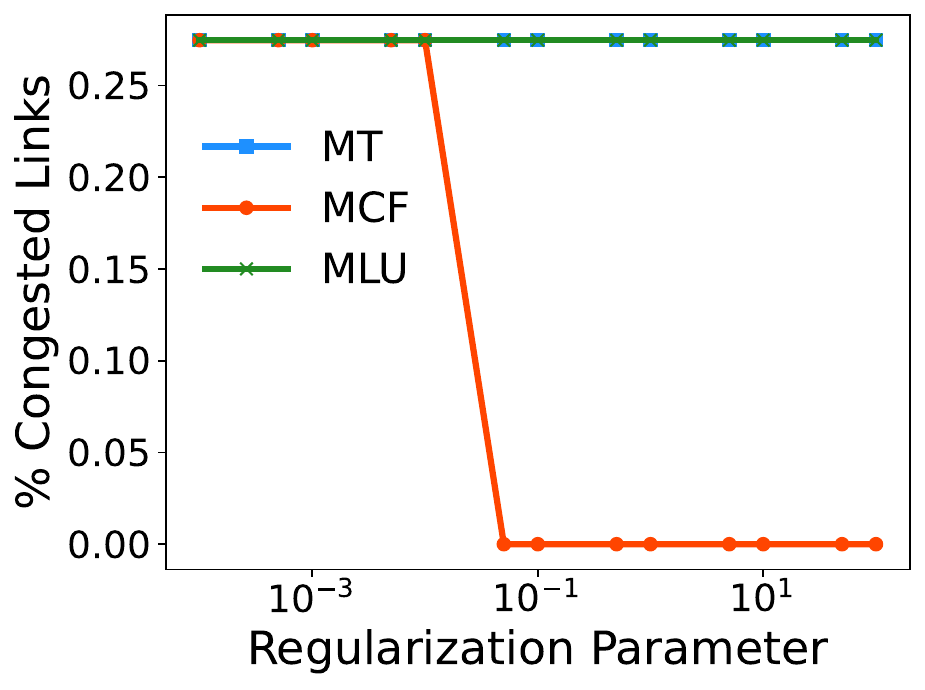}
       \caption{\small{Congestion incidence.}}
       \label{fig:reg-param-congestion}
      \end{subfigure}
   \vspace{3mm}
    \caption{Experiments varying the regularization parameter on a log scale. Throughput values in~\ref{fig:reg-param-throughput} are normalized by the maximum throughput seen over all values of $\lambda$ for each objective.}
     \label{fig:reg-param}
     \vspace{2mm}
\end{figure*}
\section{Minimizing MLU formulation}
\label{subsec:mlu-formulation}
We provide \sysname{'s} formulation for the MMLU TE objective (Algorithm~\ref{alg:min-mlu}) and discuss why even the standard LP formulation always sends full demand, even if this exceeds link capacities~\cite{smore-nsdi}. Because the objective is to minimize $Z$, without an \emph{equality} demand constraint that requires every flow to send its full request demand, the naive solution always ensures $Z$ is minimized by simply sending no traffic, \ie $w_{p} = 0$ for all paths and demands. Note that changing the demand constraint from an equality constraint to the standard inequality, $\sum_{p \in P_{i}} w_{p} \leq 1,  \forall d_{i} \in D$, would enable this default, zero-throughput solution to occur. 

If the formulation requires all demands to be fully satisfied, then it is possible that flow will exceed some link capacities without a capacity constraint. However, by enforcing capacity constraints for each link, it is possible, and very likely, that the solver will be unable to find \emph{any} feasible solution because it is impossible to satisfy all current demands with the capacity of the network. Thus, operators either use heuristics to post-process allocations computed by the MMLU objective to comply with link capacities or simply use the direct solver outputs, which are directly computed from an objective that aims to minimize congestion while sustaining full demand. In undersubscribed networks (with reasonably diverse paths), this is not a problem because MMLU can find solutions with $Z < 1$.

\begin{algorithm}[h]
    \begin{flushleft}
        \textbf{Inputs:}
    \end{flushleft}
    \begin{tabular}{p{2cm}p{6cm}}
        $G\langle V,E\rangle$ & network $G$ with vertices $V$ and links $E$\\
        $d_{i} \in D$ & traffic demand of source-destination pair $i$\\
        $P_{i}$ & set of paths for flow $i$\\
        $c_{e}$ & capacity of link $e$\\
        $I(p,e)$ & indicator variable for path $p$ using link $e$\\
        $\lambda$ & regularization parameter\\
    \end{tabular}
    \begin{flushleft}
        \textbf{Auxiliary Variables:}
    \end{flushleft}
    \vspace{-0.5em}
    \begin{tabular}{p{2cm}p{6cm}}
        $u_{e}$ & utilization of link $e$ \\
        $Z$ & maximum link utilization (MLU) \\
    \end{tabular}
    \begin{flushleft}
	    \textbf{Output:}
    \end{flushleft}
    \vspace{-0.5em}
    \begin{tabular}{p{2cm}p{6cm}}
        $w_{p}$ & flow splitting ratio along path $p$\\
     \end{tabular}
    \begin{flushleft}
        \textbf{Minimize} $Z + \lambda \sum_{e \in E} u_{e}^{2}$ \\[0.1em]
        \emph{subject to:}
    \end{flushleft}
    \begin{tabular}{ll}
        $\sum_{p \in P_{i}} w_{p} = 1$ & $\forall d_{i} \in D$\\[0.5em]
        $Z \geq u_{e}$ & $\forall e \in E$\\[0.5em]
        $u_{e} = \frac{\sum_{i} d_{i} \sum_{p \in P_{i}} w_{p} I(p,e)}{c_{e}} $ & $\forall e \in E$\\[0.5em]
        $w_{p} \geq 0$ & $\forall p \in P_{i}, \forall d_{i} \in D$\\[0.5em]
    \end{tabular}
\caption{\sysname{} MLU Minimization}
\vspace{-0.6em}
\label{alg:min-mlu}
\end{algorithm}

\section{Choice of regularization parameter}
\label{subsec:appendix-reg-param}
We examine how the choice of the regularization parameter, $\lambda$, impacts both throughput and congestion for all three objectives. We first perturb the average demand matrix $k$ times for the \meta{} network and use the same slicing strategy as in \S\ref{sec:eval-divergence}. We then run \sysnamete{} for each objective using different values of $\lambda$, ranging from $1\times10^{-4}$ to $100$. The results are shown in Figure~\ref{fig:reg-param}. We find that \sysnamete{} is very robust in ensuring low congestion regardless of the choice of $\lambda$: none of the three objectives ever congest more than 0.26\% of links, and MT and MCF consistently sustain only about 0.1\% excess flow. However, we note that the regularization parameter still matters for performance. Choosing a value of $\lambda$ that is too high can degrade performance of the original objective, as shown with maximum concurrent flow in Figures~\ref{fig:reg-param-throughput} and \ref{fig:reg-param-congestion}. Since MCF's objective value is always within the $[0,1]$ range, choosing too high a value of $\lambda$ (like $100$), can cause it to neglect the original objective and simply sustain less flow in favor of reduced congestion. Nonetheless, the choice of $\lambda$ is still insignificant within several orders of magnitude, and our evaluations across multiple WANs and thousands of demand matrices showed strong performance while fixing the value of $\lambda$ for each objective. That is, it is sufficient to simply choose a single value of $\lambda$ for each objective, regardless of the topology or demand matrix, and never modify its value again.

\section{Effective throughput on KDL}
For KDL, we also compare \sysname{}, \sysname{-Opt}, and the standard LP on the effective throughput sustained, \ie throughput minus total divergence-induced congestion (sum of all link over-allocations) and then normalized by the throughput of the \oracle{} (which cannot experience divergence-induced congestion). As shown in Figure~\ref{fig:kdl-flow}, \sysname{} sustains >99.5\% of the \oracle{'s} throughput in all cases. The equivalent performance of \sysname{} and \sysname{-Opt}  indicates that runtime optimizations do not impact performance.
\begin{figure}[H]
\centering
\includegraphics[width=0.65\linewidth]{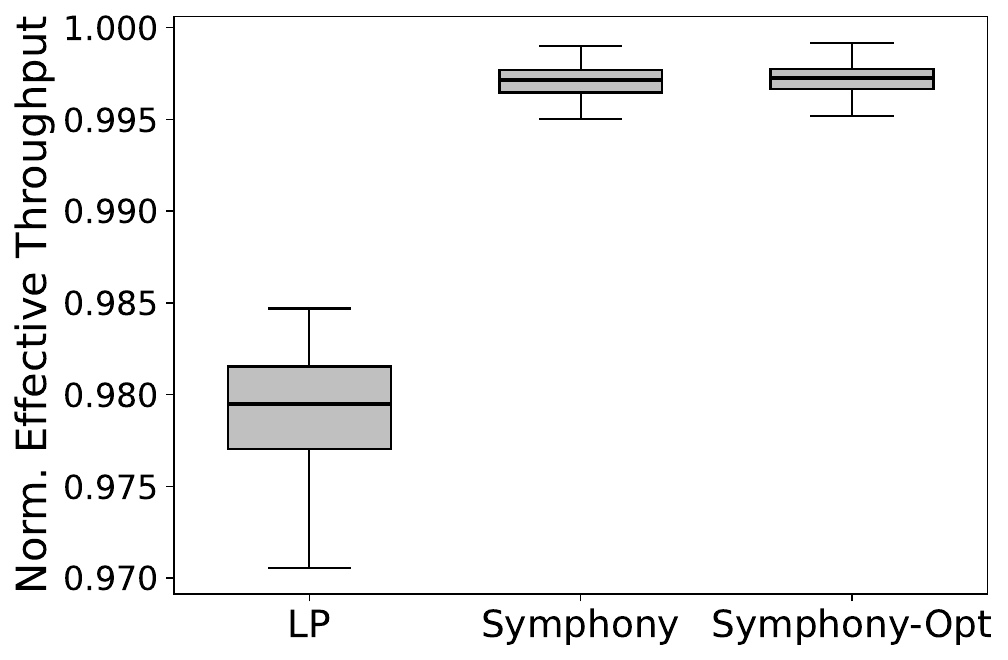}
\caption{KDL effective throughput.}
\label{fig:kdl-flow}
\end{figure}

\section{Extended \dote evaluation}
\label{subsec:dote-appendix}
We compared \sysname{} to \dote{} in two experiments: (1) average conditions, \ie randomly sampling demand matrices, and (2) worst-case conditions, \ie sampling demand matrices that are the \emph{most different}, in terms of Euclidean distance, from the average demand matrix seen while training \dote{}. We show the results from evaluating against \dote{} in average conditions for the MMLU objective in Figure~\ref{fig:min-mlu-dote}. We note that while \dote{} performs similarly to \sysname{} in terms of the distribution of the utilizations of congested links in average conditions, \dote{} underperforms \sysname{} in ``worst-case'' conditions, where inference demand matrices deviate from the training dataset. This is especially critical for decentralized TE because operators want TE objectives that are stable even in the presence of large deviations in input demands, or sudden network-wide fluctuations in predicted demands. We note that there is no similar notion of the worst-case conditions for \sysname{} because it is completely agnostic to historical demands and its flow allocations solely depend on the current demand matrix. We also note that \dote{} still requires solving LPs with a solver like Gurobi, as every input demand matrix used during training must be labeled with the optimal value of the LP solution. \sysname{} outperforms \dote{} in reducing divergence-induced congestion for the MT objective in \emph{all scenarios} and for the MMLU objective in worst-case scenarios.
\begin{figure}[h]
\centering
    \includegraphics[width=0.65\linewidth]{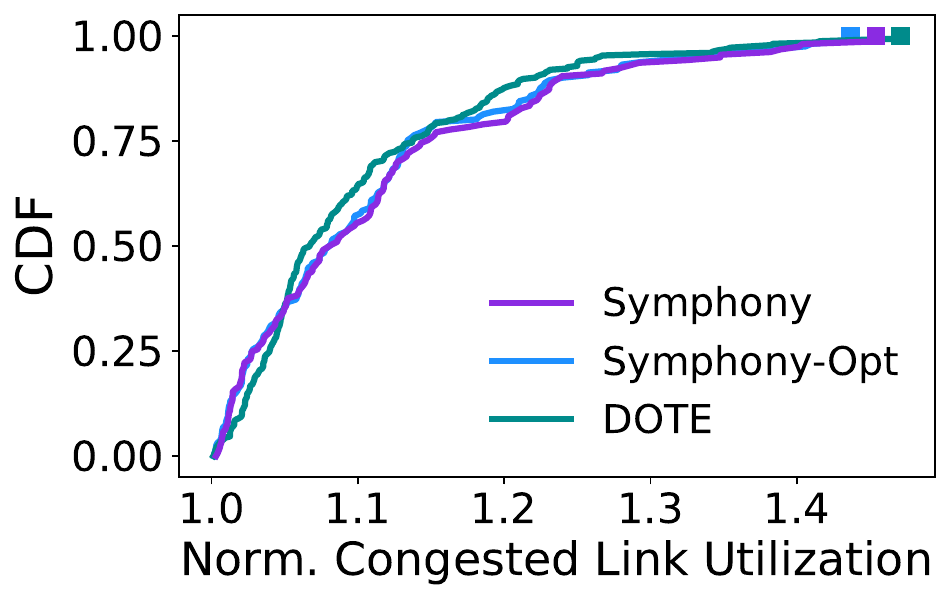}
   \caption{Comparing \sysname{} to \dote{} using a random sample of 100 demand matrices for the MMLU objective.}
   \label{fig:min-mlu-dote}
\end{figure}

\section{Randomized slicing algorithm}
\label{appendix:ilp}
Algorithm~\ref{alg:balanced_partitions} shows   
our randomized algorithm to construct slices that minimize the blast radius of 
controller faults, by aiming to load-balance egress traffic across slices (while ensuring slices are roughly equal in size).
\begin{algorithm}[t]
    \caption{\sysnameslice{}: Fault-Tolerant Partitioning}
    \label{alg:balanced_partitions}
    \begin{algorithmic}[1]
    \Require Graph $G = \langle V, E \rangle$, node weights $\phi_{v}$, number of slices $k$, desired slice sizes $s_{j}$, target per-slice total weight $T$, and tolerance $\epsilon \in [0,1]$
    \Ensure Slices $\{S_1, S_2, \dots, S_k\}$ such that each $S_j$ is connected, and $\sum_{v \in S_j} \phi_{v} \in [T \cdot (1 - \epsilon), T \cdot (1 + \epsilon)] \quad \forall j$
    \State Initialize the slices $S_j = \emptyset$, weight tracker $\theta_j = 0$, and candidates list $L_j = \emptyset$ for all $j \in \{1, \dots, k\}$
    \State Assign elephant sources to $S_j$, update $\theta_j$ with the weight, and initialize $L_j$ with neighbors of assigned nodes
    \While{any $|S_j| < s_j$}
        \For{$j = 1$ to $k$}
            \If{$|S_j| = s_j$}
                \State \textbf{Continue to next partition}
            \EndIf
            \State Prune $L_j$ to include only available nodes
            \If{$L_j = \emptyset$}
                \State \textbf{Retry or backtrack}
            \EndIf
            \State Select a node $v \in L_j$:
            \If{$\phi_v < T \cdot (1 - \epsilon)$ and $|S_j|$ is near $s_j$}
                \State Choose $v = \mathrm{arg}\max_{v' \in L_j} \phi_{v'}$ such that 
                \Statex \hspace{16mm}$\theta_j + \phi_{v'} \leq T \cdot (1 + \epsilon)$
            \Else
                \State Select a random node $v$ from $L_j$
            \EndIf
        \State Add $v$ to $S_j$, add $\phi_{v}$ to $\theta_j$, and update $L_j$ with \Statex \hspace{10mm} any unassigned neighbors of $v$
        \EndFor
    \EndWhile
    \State \Return $\{S_1, S_2, \dots, S_k\}$
    \end{algorithmic}
\end{algorithm}

\section{Evaluation methodology for blast radius}
\label{appendix:blast}
To evaluate \sysnameslice{} in Figure~\ref{fig:blast-radius}, we take the most fault-tolerant 
(\ie{} balanced weight) candidate slicing configuration generated by our algorithm and assess the worst-case percent 
of traffic that could be disrupted, using average demand distributions. (We use the slicing configuration 
generated by \sysnameslice{} because it would be infeasible to develop a new algorithm to construct slices that 
minimize blast radius using \blastshield{'s} slice routing, which explodes combinatorially in the number of 
\emph{paths} in the network, as opposed to nodes, in our case.) For \sysname{}, which 
uses source routing, this is done by first calculating the percent of total traffic that originates from nodes 
in each slice, and then taking the maximum value as the blast radius. For \blastshield{}, this 
is far more complicated, because the percent of traffic affected by a fault depends 
on the exact path through the network, which is determined by TE. Thus, to calculate the blast 
radius for \blastshield{}, we (1) for each slice, identify the set of flows associated with any 
paths that pass through the slice, (2) sum the average demand values for the flows identified 
from the first step, and (3) identify the slice with the highest total expected flow according to the 
previous step, and divide its total expected flow by the average total demand. This procedure 
captures how much traffic could potentially be affected by a controller fault using \blastshield{'s}
custom slice routing.

\section{\sysnamete{} performance on random slices}
\label{appendix:rand-slice}
We run a similar experiment as that shown in Figure~\ref{fig:balanced-slice-cdfs}, but 
instead using slicing configurations generated randomly, without using \sysnameslice{}.
The randomly generated slices still meet the same (roughly equal) size and connectivity constraints, but they do not 
aim to be fault tolerant, by balancing the weights of slices (in terms of total traffic sent 
by nodes in the slice). Figure~\ref{fig:random-slices} shows that \sysnamete{} outperforms 
all TE baselines, across all permutations of assigned demand matrices and all 100 randomly generated 
slicing configurations.

\begin{figure}[h]
\centering
 \includegraphics[width=0.65\linewidth]{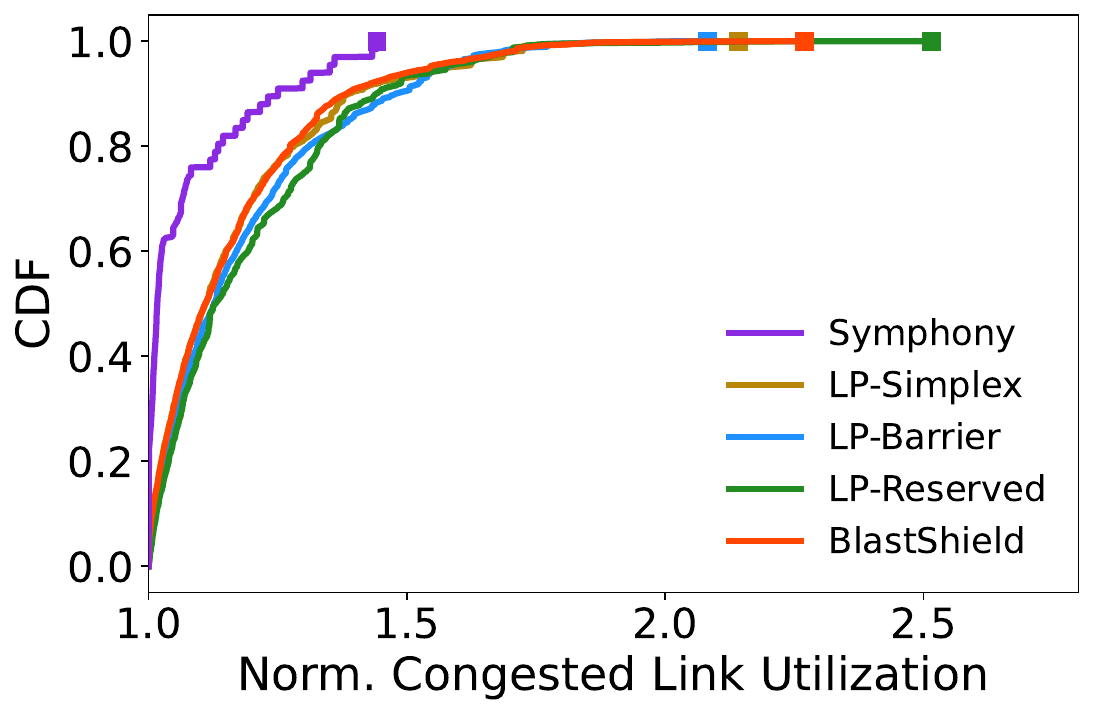}
 \caption{\sysnamete{} on randomly generated slices.}
 \label{fig:random-slices}
\end{figure}

\section{Distribution of elephant sources}
\label{appendix:elephant}
Figure~\ref{fig:heavy-hitter} illustrates the Lorenz curve of traffic generated by nodes. 
It shows that 40\% of nodes send no traffic while 20\% of nodes send over 50\% of the traffic. 
\emph{Elephant sources} dominate traffic sent over the WAN. We note that the distribution of 
traffic \emph{sent} by nodes is even more skewed than the well-documented skewed distribution in 
demands, since this definition aggregates demands by source node, amplifying the effects of 
top contributors. We use this property, which has been well documented in WAN TE, to warm-start \sysnameslice{} 
and extensively prune the search space of insufficient solutions. Including two elephant sources 
in the same slice would result in an infeasible solution, given most values of $\epsilon$ and $T$,
for \sysnameslice{} because it would immediately ``over-weight'' the slice and prevent other 
slices from achieving sufficient weight later on. By initializing slices with elephant sources, \sysnameslice{}
starts off with significant weight in every slice and proactively removes infeasible possibilities.

\begin{figure}[h]
\centering
    \includegraphics[width=0.63\linewidth]{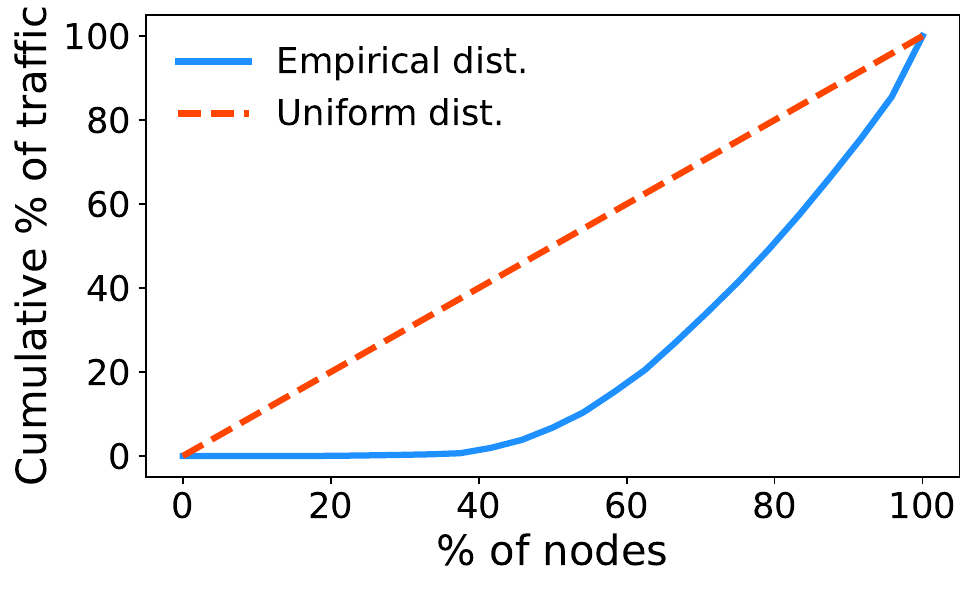}
    \caption{Elephant sources dominate mean traffic generated.}
    \label{fig:heavy-hitter}
   \end{figure}

\section{Congestion with \sysnamete{}}
\label{sec:appendix-congestion}
For the MMLU objective, the realized MLU after decentralization measures the impact of divergence.
Figure~\ref{fig:effective-mlu} shows the MLUs of different baselines, normalized by the MLU of the \oracle{} in each iteration.
\sysname{} improves median and worst-case MLUs by 21\% and 33\%, respectively for the MMLU objective.
We also show additional results for congestion across topologies, path selection strategies, and objectives in 
Figures~\ref{fig:congested-cdfs-geant}-\ref{fig:num-congested-geant}.
\begin{figure*}[h]
    \centering   
    \includegraphics[width=0.75\textwidth]{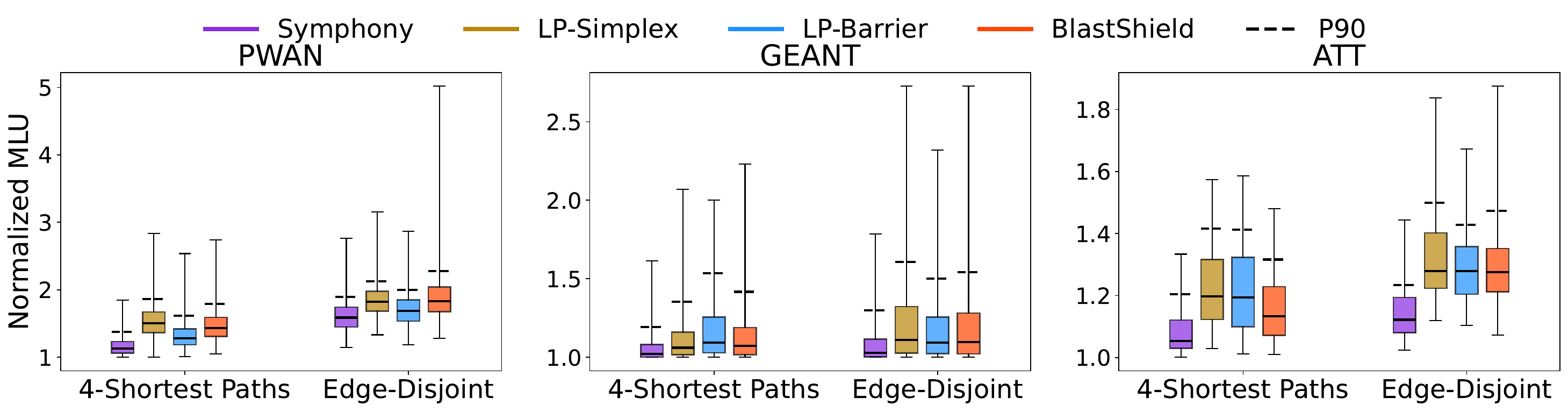}
   \vspace{2mm}
    \caption{The realized MLU, normalized by \oracle{'s} MLU value, for the minimum MLU objective.}
     \label{fig:effective-mlu}
\end{figure*}
\begin{figure*}[h]
    \centering   
    \includegraphics[width=0.75\textwidth]{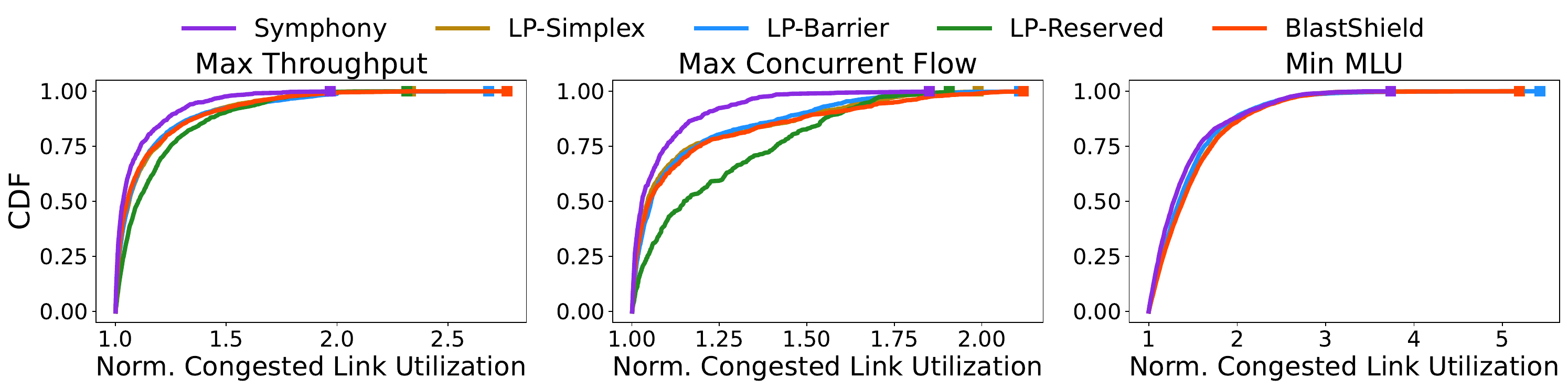}
   \vspace{2mm}
    \caption{CDFs of link utilizations for congested links, across all iterations, for GEANT with 4-shortest paths.}
     \label{fig:congested-cdfs-geant}
\end{figure*}
\begin{figure*}[h]
    \centering   
    \includegraphics[width=0.75\textwidth]{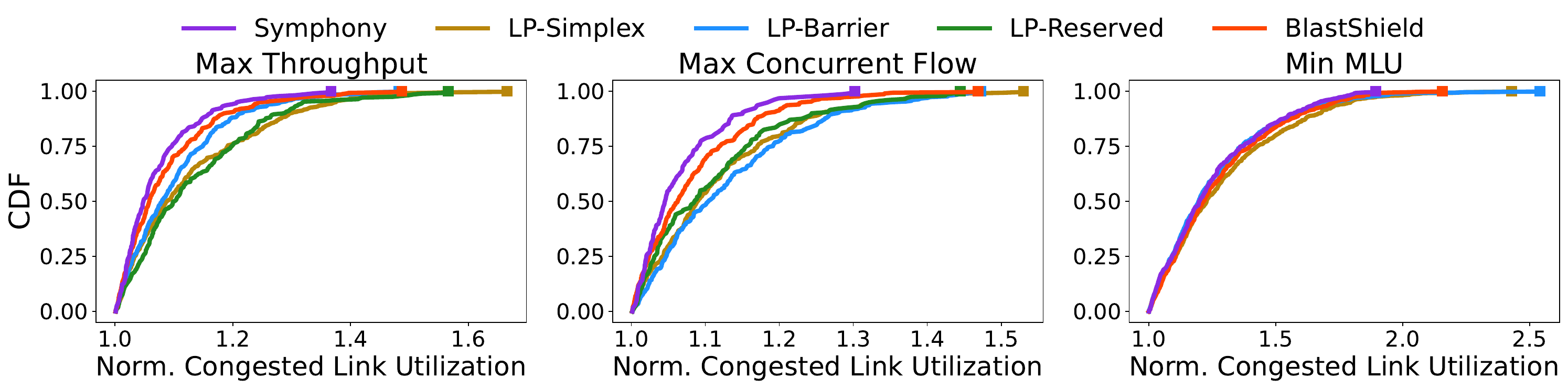}
   \vspace{2mm}
    \caption{CDFs of link utilizations for congested links, across all iterations, for ATT with 4-shortest paths.}
     \label{fig:congested-cdfs-att}
\end{figure*}
\begin{figure*}[h]
  \centering   
  \includegraphics[width=0.75\textwidth]{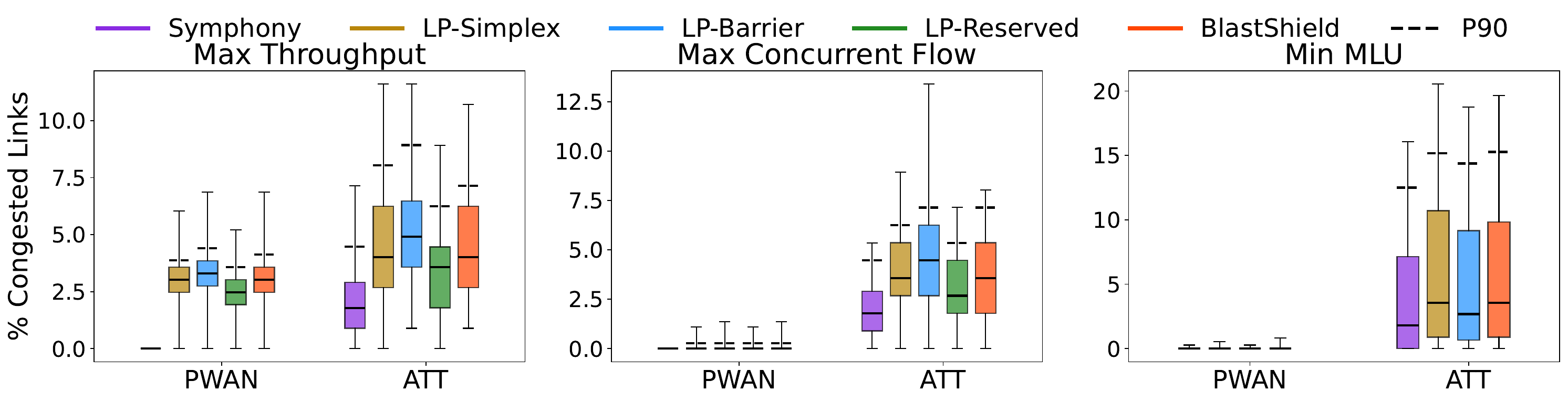}
 \vspace{2mm}
  \caption{The percent of links that are congested in each iteration, using edge-disjoint paths.}
   \label{fig:num-congested-ed}
\end{figure*}
\begin{figure*}[h]
    \centering   
    \includegraphics[width=0.75\textwidth]{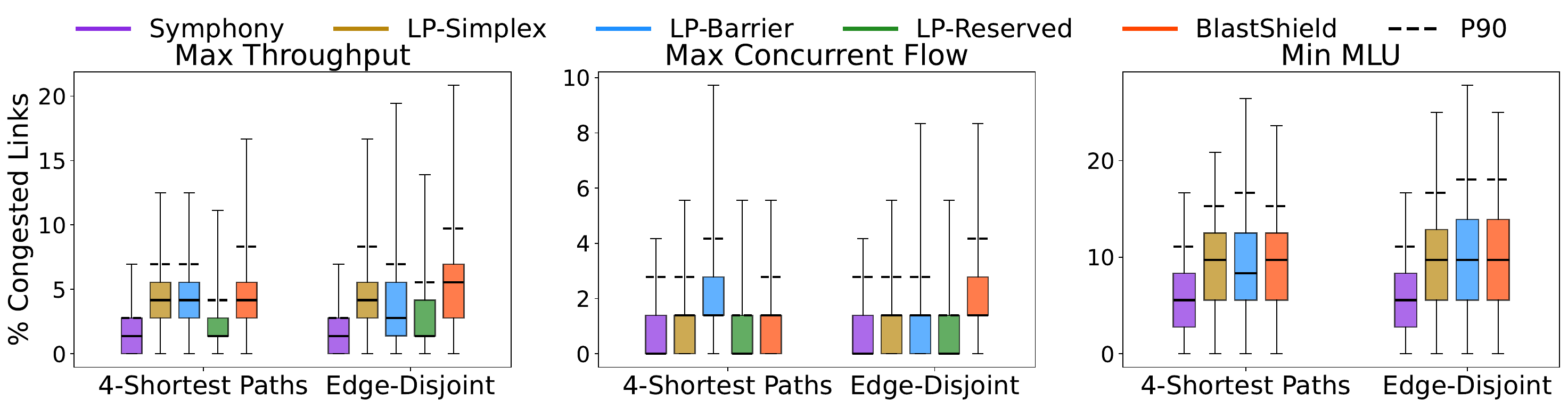}
   \vspace{2mm}
    \caption{The percent of links that are congested in each iteration for GEANT.}
     \label{fig:num-congested-geant}
  \end{figure*}
\end{document}